%% file: arxiv-secretary.tex
\documentclass[11pt]{article}

\newcommand{\journalonly}[1]{}
\newcommand{\arxivonly}[1]{#1}

\usepackage{booktabs}
\usepackage{natbib}
\usepackage{amsmath, amsthm,amssymb}
\usepackage{graphicx}

\graphicspath{{../Database/mgmt_sci_figs/}{./bayesian_figures/}}

\usepackage{url}
\usepackage{todonotes}

\newcommand{\union}{\cup}
\newcommand{\intersect}{\cap}

\usepackage{dsfont}
\newcommand{\N}{\mathds{N}}

\providecommand{\E}{\mathds{E}}\renewcommand{\E}{\mathds{E}}

\theoremstyle{plain}
\newtheorem{theorem}{Theorem}

\newtheorem{lemma}[theorem]{Lemma}

\theoremstyle{definition}
\newtheorem{definition}{Definition}

\usepackage[tight]{subfigure}
\usepackage{units}

\newcommand{\D}{\mathcal{D}}

\newcommand{\given}{\,\middle|\,}

\title{Learning in the Repeated Secretary Problem}
\author{Daniel G.\ Goldstein\\Microsoft Research \and
R.\ Preston McAfee\\Microsoft \and
\and Siddharth Suri\\Microsoft Research \and James R.\ Wright\\Microsoft Research}

\usepackage{fancyhdr}
\begin{document}

\maketitle
\begin{abstract}
In the classical secretary problem, one attempts to find the maximum of an unknown and unlearnable distribution through sequential search.  
In many real-world searches, however, distributions are not entirely unknown and can be learned through experience.
To investigate learning in such a repeated secretary problem we conduct a
large-scale behavioral experiment in which people search repeatedly from fixed distributions.  
In contrast to prior investigations that find no evidence for learning in the classical scenario, in the repeated setting we observe substantial learning resulting in near-optimal stopping behavior.  
We conduct a Bayesian comparison of multiple behavioral
models which shows that participants' behavior is best described by a class of
threshold-based models that contains the theoretically optimal strategy.   
%In a Bayesian comparison of multiple behavioral models, participants' behavior is best described by a class of threshold-based models that contains the theoretically optimal strategy.  
Fitting such a threshold-based model to data reveals players' estimated
thresholds to be surprisingly close to the optimal thresholds after only a
small number of games.    
\end{abstract}


% Text of your paper here
\input{intro.tex}
\input{related.tex}
\input{setup.tex}
\input{agg_results.tex}
\input{models_header.tex}
\input{models.tex}
\input{conclusion.tex}

\appendix
\input{appendix.tex}

% Bibliography
\clearpage
\bibliographystyle{apalike}
\bibliography{sid}

% Acknowledgments here
%\ACKNOWLEDGMENT{}

%%%%%%%%%%%%%%%%%
\end{document}

%% file: intro.tex
\section{Introduction} 

In the \emph{secretary problem},  an agent evaluates candidates one at a time in search
of the best one, making an accept or reject decision after each evaluation.
Only one candidate can be accepted, and once rejected, a candidate can never be
recalled. It is called the secretary problem because it resembles a hiring process in
which candidates are interviewed serially and, if rejected by one employer, are
quickly hired by another.

Since its appearance in the mid-twentieth century, the secretary problem 
has enjoyed exceptional popularity \citep{freeman-secretary}. 
It is the prototypical optimal stopping problem, attracting
so much interest 
from so many fields that one review article concluded that it ``constitutes its own `field' of study'' \citep{ferguson-secretary}.
\journalonly{
In this century, analyses, extensions, and tests of the secretary problem have appeared in decision science, operations research, computer science, 
economics, statistics, psychology as well as in the pages of this journal \citep{alpern-2017,palley-2014,bearden-2006b}.
}%
\arxivonly{
In this century, analyses, extensions, and tests of the secretary problem have appeared in decision science, computer science, 
economics, statistics, psychology, and operations research \citep{alpern-2017,palley-2014,bearden-2006b}.
}

The intense academic interest in the secretary problem may have to do with its
similarity to real-life search problems such as choosing a mate
\citep{todd-secretary}, choosing an apartment \citep{zwick-secretary} or
hiring, for example, a secretary. It may have to do with the way the problem
exemplifies the concerns of core branches of economics and operations research
that deal with search costs. Lastly, the secretary problem may have endured
because of curiosity about its fascinating solution. In the classic version of
the problem, the optimal strategy is to ascertain the maximum of the first
$1/e$ boxes and then stop after the next box that exceeds it. Interestingly,
this $1/e$ stopping rule wins about $1/e$ of the time in the limit
\citep{gilbert-secretary}. The curious solution to the secretary problem only holds when the values
in the boxes are drawn from an unknown distribution. To make this point clear,
in some empirical studies of the problem, participants only get to learn the
rankings of the boxes instead of the values \citep[e.g.,][]{seale-secretary}. 
But is it realistic to assume that people cannot learn about the distributions in 
which they are searching?

In many real-world searches, people can learn about the distribution of the quality of candidates as
they search. The first time a manager hires someone, she may have only a vague
guess as to the quality of the candidates that will come 
through the door. By the fiftieth hire, however, she'll have hundreds of
interviews behind her and know the distribution rather well. This should cause
her accuracy in a real-life secretary problem to increase with experience. 

While people seemingly \textit{should} be able improve at the secretary problem with
experience, surprisingly, prior academic research does not find evidence that they do. 
For example, \citet{campbell-secretary}
attempted to get participants to learn by offering enriched feedback and even
financial rewards in a repeated secretary problem, but concluded ``there is no
evidence people learn to perform better in any condition''. Similarly,
\citet{lee-secretary} found no evidence of learning, nor did
\citet{seale-secretary}.

In contrast, by way of a randomized experiment with thousands of players, we
find that performance improves dramatically over a few trials and soon
approaches optimal levels. 
%% This paper serves to document this improvement,
%% discuss why other investigations have not observed learning, and to describe
%% and model how performance improves with experience. 
%%
%% SS to DGG & JRW: Please go over the rest of this paragraph.  I'm trying
%% to show the kinds of learning we find. I also don't want to over reach.
We will show that players steadily increase their probability of
winning the game with more experience, eventually coming close to  
the optimal win rate.
Then we show that the improved win rates are due to players learning to
make better decisions on a
box-by-box basis and not just due to aggregating over boxes.
Furthermore we will show that the learning we observe occurs in a noisy environment
where the feedback they get, i.e. win or lose, may be unhelpful.
After showing various types of learning in our data we turn our attention
to modeling the players behavior.  
Using a Bayesian comparison framework we show that
%players are learning a family of threshold-based models which include the optimal strategy.
players' behavior is best described by a family of threshold-based models which include the optimal strategy.
Moreover, the estimated
thresholds are surprisingly close to the optimal thresholds after only a small number of games. 

%% file: related.tex
\section{Related Work}

While the total number of articles on the secretary problem is large
\cite{freeman-secretary}, our concern with empirical, as opposed to purely theoretical, investigations  reduces these to a much smaller set. We discuss here the most similar to our investigation. \citet{ferguson-secretary} usefully defines a ``standard'' version of the secretary problem as follows:
\begin{quote}
1. There is one secretarial position available.\\
2. The number $n$ of applicants is known. \\
3. The applicants are interviewed sequentially in random order, each order being equally likely.\\
4. It is assumed that you can rank all the applicants from best to worst without ties. The decision to accept or reject an applicant must be based \emph{only} on the relative ranks of those applicants interviewed so far.\\
5. An applicant once rejected cannot later be recalled.\\
6. You are very particular and will be satisfied with nothing but the very best.\\
\end{quote}

The one point on which we deviated from the standard problem is the
fourth. To follow this fourth assumption strictly, instead of presenting people
with raw quality values, some authors \citep[e.g.,][]{seale-secretary} present
only the ranks of the candidates, updating the ranks each time a new candidate
is inspected. This prevents people from learning about the distribution.
However, because the purpose of this work is to test for improvement when
distributions are learnable, we presented participants with actual values
instead of ranks.

Others properties of the classical secretary problem could have been changed.
For example, there exist alternate versions in which there is a payout for choosing candidates
other than the best. These ``cardinal'' and ``rank-dependent'' payoff
variants \citep{bearden-2006} violate the sixth property above.
We performed a literature search and found fewer than 100 papers on these
variants, while finding over 2,000 papers on the standard variant.
Our design preserves the sixth property for two reasons.
First, by preserving it, our results will be directly comparable to the
the greatest number of existing theoretical and empirical
analyses. 
Second, changing more than one variable at a time is undesirable because it 
makes it difficult to identify which variable change is responsible for
changes in outcomes.

While prior investigations, listed next, have looked at people's performance on
the secretary problem, none have exactly isolated the condition of making the
distributions learnable. Across several articles, Lee and colleagues
\citep{lee-secretary, campbell-secretary, lee-oconnor-secretary} conducted
experiments in which participants were shown values one at a time and were told
to try to stop at the maximum. Across these papers, the number of candidates
(or boxes or secretaries) ranged from 5 to 50 and participants played from 40
to 120 times each. In all these studies, participants knew that the values were
drawn from a uniform distribution between 0 and 100. For instance,
\citet{lee-secretary} states, ``It was emphasized that ...  the values were
uniformly and randomly distributed between 0.00 and 100.00''. With such an
instruction, players can immediately and exactly infer the percentiles of the
values presented to them, which helps them calculate the probability that
unexplored values may exceed what they have seen. 
As participants were told about the distribution, these experiments do not
involve learning distribution from experience, which is our concern.
Information about the
distribution was also conveyed to participants in a study by
\citet{rapoport-secretary-1970}, in which seven individual participants viewed
an impressive 15,600 draws from probability distributions over several weeks
before playing secretary problem games with values drawn from the same
distributions. These investigations are similar to ours in that they both
involve repeated play and that they present players with actual values instead
of ranks. That is, they depart from the fourth feature of the standard
secretary problem listed above. These studies, however, differ from ours in that
they give participants information about the distribution from which the values
are drawn before they begin to play. In contrast, in our version of the game,
participants are told no information about the distribution, see no samples
from it before playing, and do not know what the minimum or maximum values
could be. This key difference between the settings may have had a great impact.
For instance, in the studies by Lee and colleagues, the authors did not find
evidence of learning or players becoming better with experience. In contrast,
we find profound learning and improvement with repeated play.

\citet{corbin-secretary} ran an experiment in which people played repeated
secretary problems, with a key difference that these authors manipulated the
values presented to subjects with each trial. For instance, the authors varied
the support of the distribution from which values were drawn, and manipulated
the ratio and ranking of early values relative to later ones. The manipulations
were done in an attempt to prevent participants from learning about the
distribution and thus make each trial like the ``standard'' secretary problem
with an unknown distribution. Similarly, \citet{palley-2014} provide participants
with ranks for all but the selected option to hinder learning about the 
distribution. In contrast, because our objective is to
investigate learning, we draw random numbers without any manipulation.

Finally, in a study by \citet{kahan-decision}, groups of 22 participants
were shown up to 200 numbers chosen from either a left skewed, right skewed
or uniform distribution.  In this study, as well as ours, participants were
presented with actual values instead of ranks. Also
like our study, distributions of varying skew were used as
stimuli. However, in \citet{kahan-decision}, participants played the game
just one time and thus were not able to learn about the distribution
to improve at the game.  

In summation, for various reasons, prior empirical investigations of the
secretary problem have not been designed to study learning about the
distribution of values. These studies either informed participants about the
parameters of the distribution before the experiment, allowed participants
to sample from the distribution before the experiment, replaced values from
the distribution with ranks, manipulated values to prevent learning, or ran
single-shot games in which the effects of learning could not be applied to
future games. Our investigation concerns a repeated secretary problem in
which players can observe values drawn from distributions that are held 
constant for each player from game to game.

%% file: setup.tex
\section{Experimental setup}

To collect behavioral data on the repeated secretary problem with learnable
distributions of values, we created an online experiment. The experiment was
promoted as a contest on several web logs and attracted 6,537 players who
played the game at least one time. A total of 48,336 games were played on the
site. As users arrived at the game's landing page, they were cookied and their
browser URL was automatically modified to include an identifier. These two
steps were taken to assign all plays on the same browser to the same user id
and condition, and to track person-to-person sharing of the game URL. Any user
determined to arrive at the site via a shared URL (i.e., a non-cookied user
entering via a modified URL) was excluded from analysis and is not counted in
the 6,537 we analyze. We note that including these users makes little
difference to our results and that we only exclude them to obtain a set of
players that were randomly assigned to conditions by the website. Users saw the
following instructions. Blanks stand in the place of the number of boxes, which
was randomly assigned and will be described later. 

\begin{quote}
You have been captured by an evil dictator. 
He forces you to play a game. 
There are \underline{\hspace{.5cm}} boxes. 
Each box has a different amount of money in it. 
You can open any number of boxes in any order. 
After opening each box, you can decide to open another box or you can stop by clicking the stop sign. 
If you hit stop right after opening the box with the most money in it (of the \underline{\hspace{.5cm}} boxes), then you win. 
However, if you hit stop at any other time, you lose and the evil dictator will kill you. 
Try playing a few times and see if you improve with practice.
\end{quote}

The secretary problem was lightly disguised as the ``evil dictator game'' to
someone lessen the chances that a respondent would search for the problem online
and discover the classical solution. 
Immediately beneath the instructions was an icon of a traffic stop sign and the
message ``When you are done opening boxes, click here to find out if you win''.
Beneath this on the page were hyperlinks stating ``Click here to open the first
box'',``Click here to open the second box'', and so on. As each link was
clicked, an AJAX call retrieved a box value from the server, recorded it in a
database and presented it to the user. If the value in the box was the highest
seen thus far, it was marked as such on the screen.  See Figure
\ref{fig:screenshot} in the Appendix for a screenshot. Every click and box
value was recorded, providing a record of every box value seen by every player,
as well as every stopping point. If a participant tried to stop at a box that
was dominated by (i.e., less than) an already opened box, a pop-up explained
that doing so would necessarily result in the player losing. After clicking on
the stop icon or reaching the last box in the sequence, participants were
redirected to a page that told them whether they won or lost, and showed them
the contents of all the boxes, where they stopped, where the maximum value was,
and by how many dollars (if any) they were short of the maximum value. To
increase the amount of data submitted per person, players were told ``Please
play at least six times so we can calculate your stats''. 

\begin{figure}[t]
\centering
\includegraphics[width=0.5\textwidth]{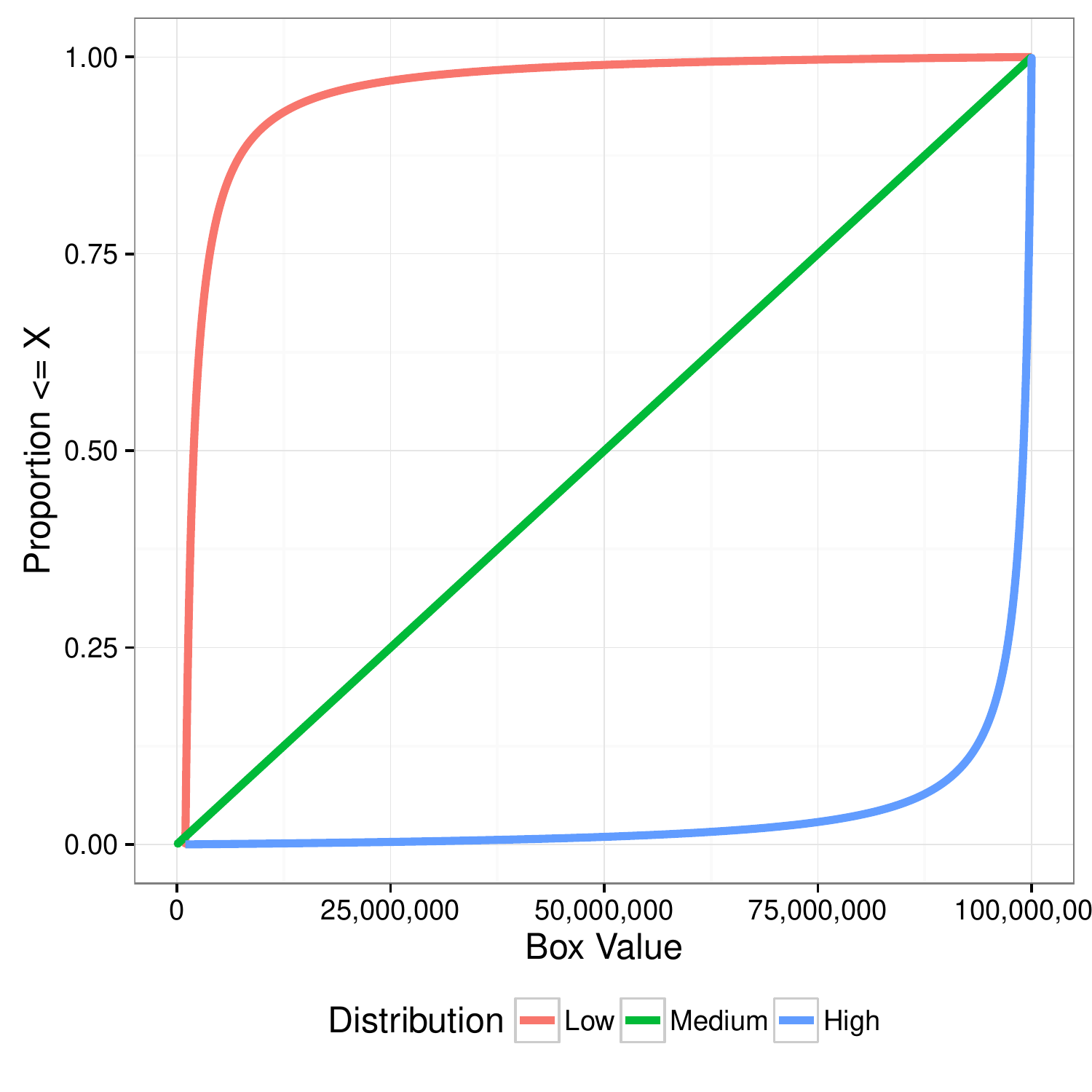}
%\vspace{-1\baselineskip}
\caption{
    Cumulative distribution functions (CDFs) of the three distributions from which box values were randomly drawn in the experiment. As probability density functions, the ``low'' distribution is strongly positively skewed, the ``medium'' distribution is a uniform distribution, and the ``high'' distribution is strongly negatively skewed.
}
%\vspace{-1\baselineskip}
\label{fig:drawn_dists}
\end{figure}

\subsection{Experimental Conditions} To allow for robust conclusions that are
not tied to the particularities of one variant of the game, we randomly varied
two parameters of the game: the distributions and the number of boxes. Each
player was tied to these randomly assigned conditions so that their immediate
repeat plays, if any, would be in the same conditions.
\subsubsection{Random assignment to distributions} The box values were randomly
drawn from one of three probability distributions, as pictured in
Figure~\ref{fig:drawn_dists}. The maximum box value was 100 million, though
this was not known by the participants. The ``low'' condition was strongly
negatively skewed. Random draws from it tend to be less than 10 million, and
the maximum value tends to be notably different than the next highest value.
For instance, among 15 boxes drawn from this distribution, the highest box
value is, on average, about 14.5 million dollars higher than the second highest
value. In the ``medium'' condition numbers were randomly drawn from a uniform
distribution ranging from 0 to 100 million. The maximum box values in 15 box
games are on average 6.2 million dollars higher than the next highest values.
Finally, in the ``high'' condition, boxes values were strongly negatively
skewed, and bunched up near 100 million. In this condition, most of the box
values tend to look quite similar (typically eight-digit numbers greater than
98 million). Among 15 boxes, the average difference between the maximum value
and the next highest is rather small at only about 80,000 dollars. Note that
the distributions are merely window dressing and are irrelevant for playing the
game. Players only need to attend to percentiles of the distribution to make
optimal stopping decisions. However, the varying distributions leads to more
generalizable results than an analysis of a single, arbitrary setting.

\subsubsection{Random assignment to number of boxes} The second level of random
assignment concerned the number of boxes, which was either 7 or 15. While one
would think this approximate doubling in the number of boxes would make the
game quite a bit harder, it only affects theoretically optimal win rates by
about 2 percentage points, as will be shown. Like with the distributions,
varying the number of boxes leads to more generalizable results.

%  skew_factor num_boxes avg_max_med_diff avg_max_2nd_diff
%2           1        15         23682762       14,519,257
%4           2        15         43700555        6,245,740
%6           3        15          1029317          79,654
%1           1         7         14371621       10,657,392
%3           2         7         37530617       12,579,616
%5           3         7          1110373         225,298

With either 7 or 15 boxes and three possible distributions, the experiment had
a $2 \times 3$ design. In the 7 box condition, 1145, 1082, and 1103 participants
were randomly assigned to the low, medium, and high distributions,
respectively, and in the 15 box condition, the counts were 1065, 1127, and 1015,
respectively. %% The differences in cell counts were non significant ($p = .08$, by
%% chi-square test), consistent with successful random assignment. We next turn to
%% describing how an optimal agent would go about playing the game.
%% %% SS to DG: I need your help in figuring out exactly which chi-squared
%% %% test we want to report.
%%
%% The number of subjects assigned to the two box conditions (7 or 15), the three
%% distributions of box values (low, medium, or high), and all six cells of the experiment
%% were not significantly different from each other (all $p$s $> 0.05$) by
%% chi-squared tests.
%%
The number of subjects assigned to the different conditions was not
signficiantly different whether comparing the two box conditions (7 or 15), the three
distributions of box values (low, medium, or high), or all six cells of the experiment
by chi-squared tests (all $p$s $> 0.05$).

\subsection{Optimal play} Before we begin to analyze the behavioral data
gathered from these experiments we first discuss how one would play this game
optimally. Assume values $X_{t}$ are drawn in an independently and identically
distributed fashion from a cumulative distribution function $F$. One period
consists of a player opening a box with a realization $x_{t}$ of $X_{t}$ in box
$t$. Periods are numbered in reverse order starting at $T$, so $t
= T, \ldots, 1$. Periods are thus numbered in the reverse order of the boxes
in the game, that is, opening the first box implies being in the seventh period (of a seven box game).
An action is to select or reject. For example at time $t$,
select box $t$, otherwise reject box $t$. Let \[h_{t} =\max \; \, \left\{x_{t}
,...,x_{T} \right\}.\] The history summary $h_{t+1}$ is visible to the player
at time $t$. The payoff of the player who selects box $t$ is
\[\left\{\begin{array}{cc} {1} & {h_{t} =h_{1} } \\ {0} & {h_{t} <h_{1} }
\end{array}\right. .\] Thus, the player only wins when they select the highest
value.  The problem is nontrivial because they are forced to choose without
knowing future realizations of the $X_i$.

Optimal players will adopt a threshold rule, which says, possibly as a function of the
history, accept the current value if it is greater than a critical dollar value
$c_{t}$.  It is a dominant strategy to reject any realization worse than the
best historically observed, except for the last box which must be accepted if
opened.  In addition, in our game, a pop-up warning prevented players from
choosing dominated boxes.

With a known distribution independently distributed across periods, the
critical dollar value will be the maximum of the historically best value and a
critical dollar  value that does not depend on the history.  The reason that the
critical dollar value does not depend on the history is that there is nothing to
learn about the future (known distribution) so it is either better to
accept the current value than wait for a better future value or not; the
point of indifference is exactly our critical dollar value.  Thus, the threshold
comes in the form $\max \{c_t, h_t\}$.

In addition, $c_{t}$ is non-decreasing in $t$.  Suppose, for the sake of
contradiction, that $c_{t} > c_{t+1}$.  Taking any candidate in the interval
$(c_{t+1}, c_{t})$ entails accepting a candidate and then immediately
regretting it, because as soon as the candidate is accepted, the candidate is
no longer acceptable, being worse than $c_{t}$.

Let $i = T - t,$ and $z_{i} = F(c_{t})$. We refer to the $z_{t}$ as the
critical values, which are are the probabilities of observing a value less than
the critical dollar values $c_{t}$. Let $p_{t}(h)$ be the probability of a win
given a history $h$. Table \ref{tbl:1} provides the critical values $z_{t}$ and
the probability of winning given a zero history $p_{t}(0)$ for fifteen periods.
Derivations of these figures are found in section \ref{sec:comp_cv_pwin} in the
Appendix.

\begin{table}\centering
  \caption{Critical values and probability of
    winning given a known distribution of values for up to 15 boxes.}
  \label{tbl:1}
  \begin{tabular}{p{0.35in}|p{0.55in}|p{0.55in}}
     Boxes left, $t$ & Critical values $z_{t}$ & Pr(Win) $p_{t}(0)$ \\ \hline 
  1 & 0       &    1   \\       
  2 & 0.5     & 0.750  \\     
  3 & 0.6899  & 0.684  \\     
  4 & 0.7758  & 0.655  \\   
  5 & 0.8246  & 0.639  \\   
  6 & 0.8559  & 0.629  \\   
  7 & 0.8778  & 0.622  \\   
  8 & 0.8939  & 0.616  \\   
  9 & 0.9063  & 0.612  \\   
  10 & 0.9160 & 0.609  \\   
  11 & 0.9240 & 0.606  \\   
  12 & 0.9305 & 0.604  \\   
  13 & 0.9361 & 0.602  \\   
  14 & 0.9408 & 0.600  \\   
  15 & 0.9448 & 0.599 
  \end{tabular}
\end{table}

The relevant entries for our study are the games of 7 and 15 periods.
These calculations, which coincide with those found in
\citet{gilbert-secretary},
who did not provide Equation
\eqref{GrindEQ__4_}
(see the Appendix), show that experienced players, who know the
distribution, can hope to win at best 62.2\% of the games for 7 period
games and just under 60\% of the time for 15 period games.  Note that these numbers compare favorably with the usual secretary result, which are lesser  for all game lengths, converging to the famous $1/e$ as the length diverges.
Thus there is substantial value in knowing the distribution.

As is reasonably well known, the value of the classical secretary solution
can be found by choosing a value $k$ to sample, and then setting the best
value observed in the first $k$ periods as a critical value.  The
distribution of the maximum of the first $k$ is $F(x)^{k}$.  The
probability that a better value is observed in round $m$ is $(1 - F(x)) F(x)^{m - k - 1}$.  Suppose this value is $y$; then this value wins with probability $F(y)^{T-m}$.  Thus the probability of winning for a fixed value of $k$ and $T$ periods is ${\tfrac{k}{T}} \sum_{m=k+1}^{T}{\tfrac{1}{m-1}}$.
See the derivation in Section \ref{sec:comp_pwin_unknown} in the Appendix.
The optimal value of $k$ maximizes $\frac{k}{T} \sum
_{m=k+1}^{T}\frac{1}{m-1}  $ and is readily computed to yield
Table~\ref{tbl:2}.  
%% SS to DGG & RJW: Is the following comparison correct?  Do the rows in
%% each table refer to the same thing?  If so, should we put the tables
%% side by side?  (I have code commented out below that does that, if need be)
Comparing the probability of winning shown in
Tables~\ref{tbl:1} and~\ref{tbl:2} shows that making the distribution
learnable allows for a much higher rate of winning.

\begin{table}\centering
\caption{The probability of winning a game in the classical secretary
problem (unknown distribution of values) for up to fifteen boxes}
\label{tbl:2}
 \begin{tabular}{c|c}
Game Length & Classical Secretary \\ 
(Periods)   & Problem Pr(Win) \\ \hline 
 1 &   1 \\       
 2 &   0.50  \\     
 3 &   0.50  \\     
 4 &   0.458 \\   
 5 &   0.433 \\   
 6 &   0.428 \\   
 7 &   0.414 \\   
 8 &   0.410 \\   
 9 &   0.406 \\   
 10 &  0.399 \\   
 11 &  0.398 \\   
 12 &  0.396 \\   
 13 &  0.392 \\   
 14 &  0.392 \\   
 15 &  0.389
\end{tabular}
\end{table}

%% SS: If we want a side by side version of our tables we want to comment
%% this in.
%% \begin{table}
%% \parbox{.45\linewidth}{
%%   \centering
%%   \begin{tabular}{p{0.35in}|p{0.55in}|p{0.55in}}
%%      $t$ & Critical values $z_{t}$ & Pr(Win) $p_{t}(0)$ \\ \hline 
%%   1 & 0       &    1   \\       
%%   2 & 0.5     & 0.750  \\     
%%   3 & 0.6899  & 0.684  \\     
%%   4 & 0.7758  & 0.655  \\   
%%   5 & 0.8246  & 0.639  \\   
%%   6 & 0.8559  & 0.629  \\   
%%   7 & 0.8778  & 0.622  \\   
%%   8 & 0.8939  & 0.616  \\   
%%   9 & 0.9063  & 0.612  \\   
%%   10 & 0.9160 & 0.609  \\   
%%   11 & 0.9240 & 0.606  \\   
%%   12 & 0.9305 & 0.604  \\   
%%   13 & 0.9361 & 0.602  \\   
%%   14 & 0.9408 & 0.600  \\   
%%   15 & 0.9448 & 0.599 
%%   \end{tabular}
%% \caption{Critical values and probability of winning given a known
%%   distribution of values for up to 15 boxes.}
%% }
%% \hfill
%% \parbox{.45\linewidth}{
%%   \centering
%%   \begin{tabular}{c|c}
%%    Game Length & Classical Secretary \\ 
%%    (Periods)   & Problem Pr(Win) \\ \hline 
%%    1 &   1 \\       
%%    2 &   0.50  \\     
%%    3 &   0.50  \\     
%%    4 &   0.458 \\   
%%    5 &   0.433 \\   
%%    6 &   0.428 \\   
%%    7 &   0.414 \\   
%%    8 &   0.410 \\   
%%    9 &   0.406 \\   
%%    10 &  0.399 \\   
%%    11 &  0.398 \\   
%%    12 &  0.396 \\   
%%    13 &  0.392 \\   
%%    14 &  0.392 \\   
%%    15 &  0.389
%%  \end{tabular}
%% \caption{The probability of winning a game in the classical secretary
%%   problem (unknown distribution of values) for up to fifteen boxes}
%% } 
%% \end{table}

How well can players do \textit{learning} the distribution? To model this, we
consider an idealized agent that plays the secretary problem repeatedly 
and learns from experience.
The agent begins with the critical
values and learns the percentiles of the distribution from experience; it will
be referred to as the ``LP'' (learn percentiles) agent. 
The agent has a perfect memory, makes no mistakes, has derived the critical
values in Table \ref{tbl:1} correctly, and can re-estimate the percentiles of
a distribution with each new value it observes. It is difficult to imagine a 
human player being able to learn at a faster rate than the LP agent. We thus
include it as an unusually strong benchmark.

\subsection{Learning Percentiles: The LP agent} 
\label{sec:lpagent}
The LP agent starts off knowing
the critical values for a 7 or 15 box game in percentile terms. To be precise,
these critical values are the first 7 or 15 rows under the heading  $z_{t}$ in
Table~\ref{tbl:1}. (Despite the term ``percentile'', we use decimal notation
instead of percentages for convenience.) The reason that the LP agent is not
given the critical values as raw box values is that these would be unknowable
because the distribution is unknown before the first play. However, it is
possible to compute these critical values as percentiles from first principles,
as we have done earlier in this section and in the Appendix.
Armed with these critical values, the LP agent converts the box values it
observes into percentiles in order to compare them to the critical values. The
first box value the LP agent sees gets assigned an estimated percentile of .50.
If the second observed box value is greater than the first, it estimates the
second value's percentile to be .75 and re-estimates the first value's
percentile to be .25. If the second value is smaller than the first, it assigns
the estimate of .25 to the second value and .75 to the first value. It
continues in this way, re-estimating percentiles for every subsequent box value
encountered according to the percentile rank formula: 
\begin{equation}
\frac{N_{<} + 0.5 N_{=}}{N}
    \label{eqn:pertile}
\end{equation}
where $N_{<}$ is the number of values seen so far that are less
than the given value, $N_{=}$ is the number of times the given value has occurred
so far, and $N$ is the number of boxes opened so far.

After recomputing all of the percentiles, the agent compares the percentile of
the box just opened to the relevant critical value and decides to stop if the
percentile exceeds the critical value, or decides to continue searching if it
falls beneath it, making sure never to stop on a dominated box unless it is in
the last position and therefore has no choice. Recall that a dominated box is
one that is less than the historical maximum in the current game. The encountered
values are retained from game to game, meaning that
the agent's estimates of the percentiles of the distribution will approach perfection and win rates
will approach the optima in Table~\ref{tbl:1}.

\begin{figure}[t]
\centering
\includegraphics[width=0.90\textwidth]{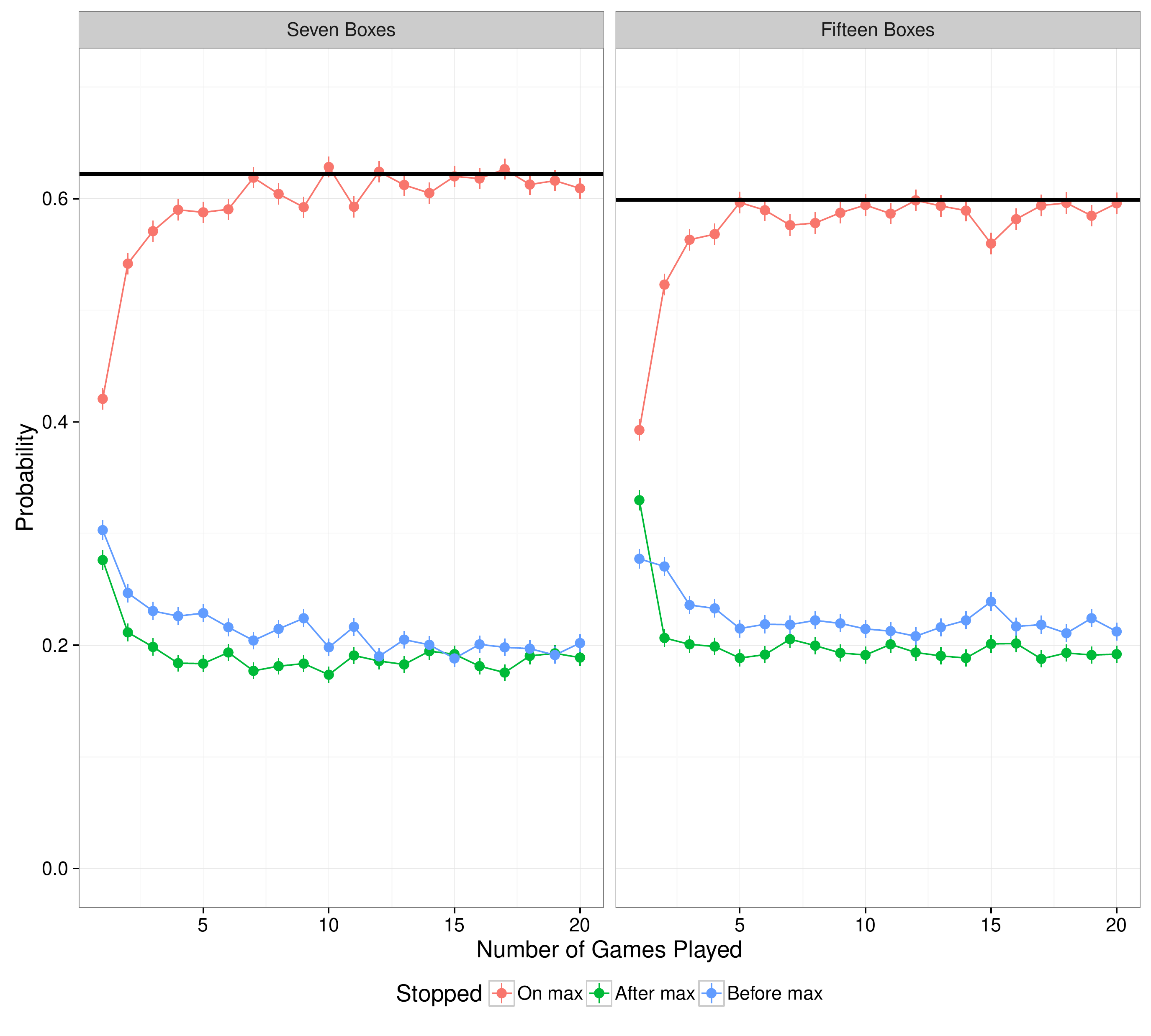}
%\vspace{-1.5\baselineskip}
\caption{
    Rates of winning (red lines), stopping too soon (blue lines) and
    stopping too late (green lines) for the LP agent. The theoretically maximal win rates for 7 and 15
    boxes are given by the solid black lines.
}
%\vspace{-1\baselineskip}
\label{fig:simulation_win_errors}
\end{figure}

How well does the LP agent perform? 
Figure \ref{fig:simulation_win_errors} shows its performance. 
Comparing its win rate on the first play to the 7 and 15 box entries in Table~\ref{tbl:2}, we see that the LP agent matches the performance of the optimal player of the classic secretary problem in its first game. 
Performance increases steeply over the first three games and achieves the
theoretical maxima (black lines) in seven or fewer games.
In any given game a player can either stop when it sees the maximum value, in which case it wins, or the player could stop before or after the maximum value, in which case it loses.  
In addition to the win rates, Figure \ref{fig:simulation_win_errors} also shows how often agents commit these two types of errors. Combined error is necessarily the complement of the win rate so the steep gain in one implies a steep drop the other. Both agents are more likely to stop before the maximum as opposed to after it, which we will see is also the case with human players.  

The LP agent serves as strong benchmark against which human performance can be
compared. It is useful to study its performance in simulation because the
existing literature provides optimal win rates for many variations of the
secretary problem, but is silent on how well an idealized agent would do when
learning from scratch. In addition to win rates, these agents show the patterns
of error that even idealized players would make on the path to optimality. In
the next section, we will see how these idealized win and error rates compare to  
those of the human players in the experiment. 

%% file: agg_results.tex
\section{Behavioral Results: Learning Effects}

As 48,336 games were played by 6,537 users, the average user played 7.39
games. Roughly half (49.6\%) of users played 5 games or more, a quarter
(23.2\%) played 9 games or more, and a tenth (9.3\%) played 16 games or more.

\begin{figure}[t]
\centering
\includegraphics[width=0.9\textwidth]{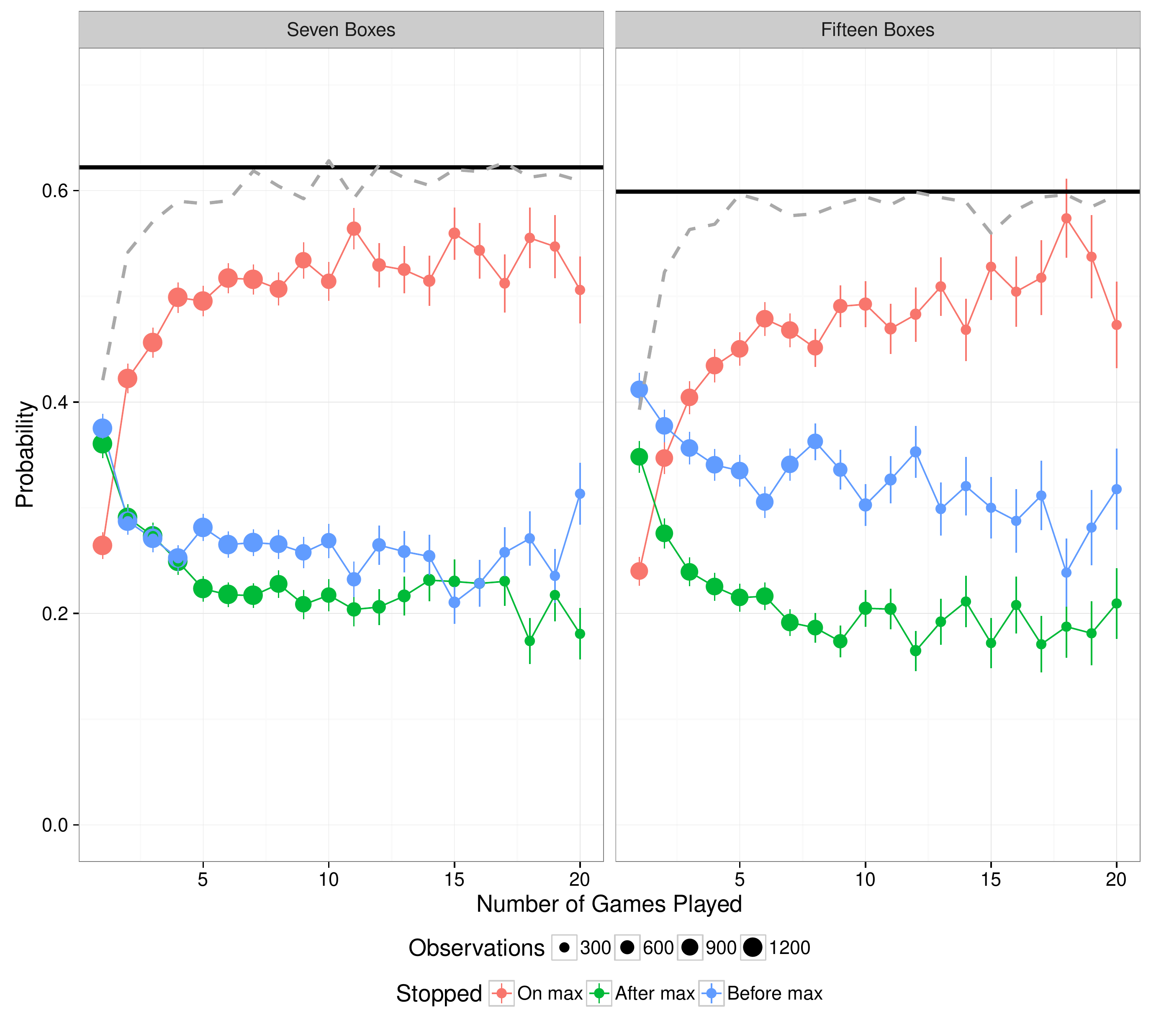}
%\vspace{-1\baselineskip}
\caption{Solid red, blue and green lines indicate the rates of winning the game and committing errors for human players with varying levels of experience. Error
  bars indicate $\pm 1$ standard error; when they are not visible they are
smaller than the points.
The area of each point is proportional to the number of players in the average.
The graph is cut at 20 games as less than 1\% of games played were beyond a
user's 20th.  The dashed gray line is the rate of winning the game for the
LP agent.  The solid black horizontal lines indicate the maximal win rate
for an agent with perfect knowledge of the distribution.
}
%\vspace{-0.75\baselineskip}
\label{fig:prob_win}
\end{figure}

\begin{figure}
\centering
\includegraphics[width=0.9\textwidth]{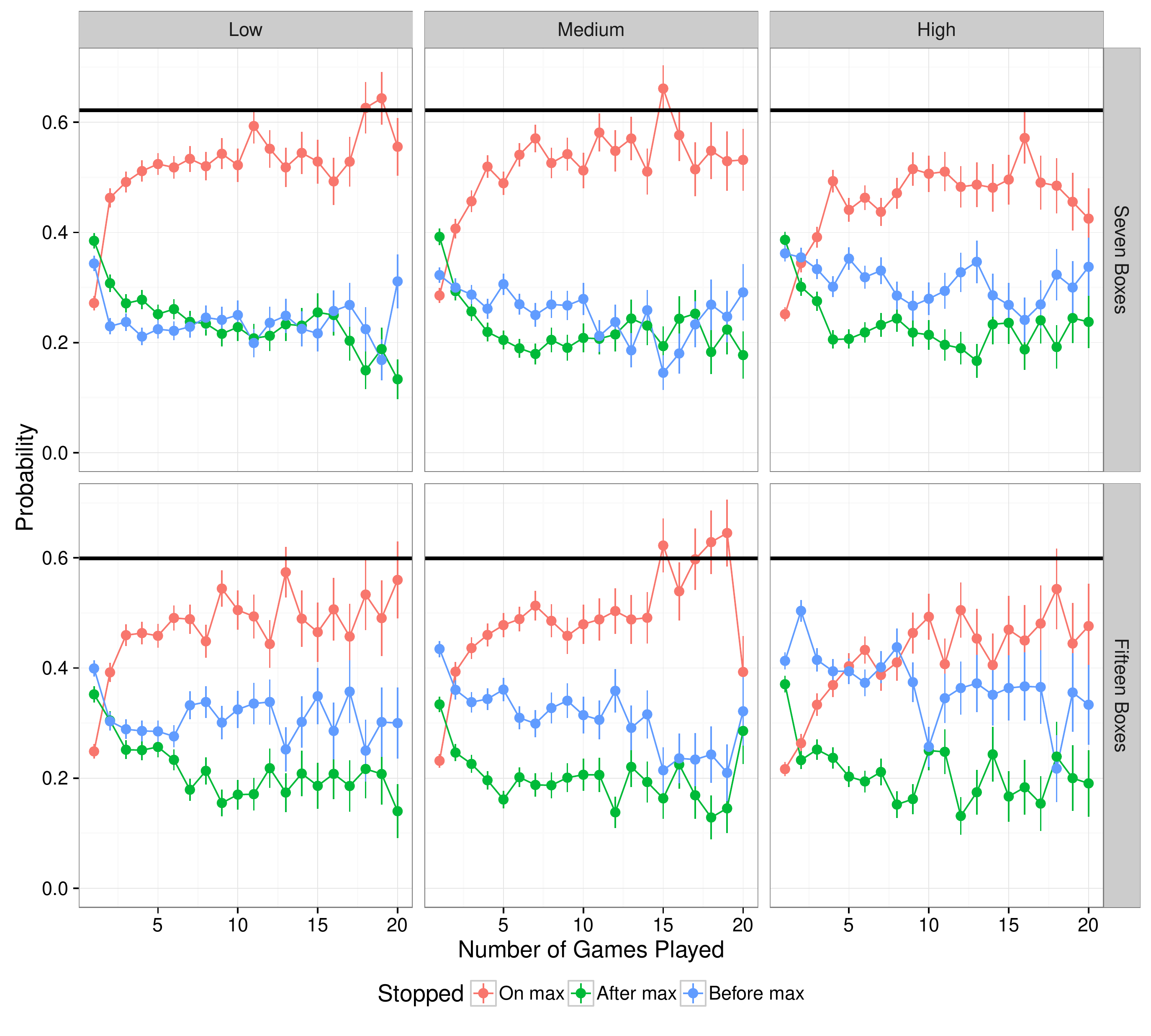}
\caption{Solid red, blue, and green lines indicate the rates of winning the game and committing errors for human players with varying levels of experience. Error
  bars indicate $\pm 1$ standard error; when they are not visible they are
smaller than the points. The solid black horizontal lines indicate the maximal win rate
for an agent with perfect knowledge of the distribution.}
\label{fig:prob_win6}
\end{figure}

Prior research \citep[e.g.,][]{lee-secretary} has found no evidence of learning
in repeated secretary problems with known distributions.  What happens with
unknown but learnable distribution? As shown in Figure~\ref{fig:prob_win},
players rapidly improve in their first games and come within five to ten
percentage points of theoretically maximal levels of performance. The leftmost
point on each red curve indicates the how often  first games are won.  The next
point to the right represents second games, and so on.  The solid black lines
at .622 and .599 show the maximal win rate attainable by an agent with perfect
knowledge of the distribution.  Note that these lines are not a fair comparison
for early plays of the game in which knowledge of the distribution is imperfect
or completely absent; in pursuit of a fair benchmark, we computed the
win rates of the idealized LP agent shown in the dashed gray lines.

Performance in the first games, in which players have very little knowledge of
the distribution is quite a bit lower than would be expected by optimal play in
the classic secretary problem with 7 (optimal win rate .41) or 15 boxes
(optimal win rate .39). Thus, some of the learning has to do with starting
from a low base. However, the classic version's optima are reached by about the
second game and improvement continues another 10 to 15 percentage points beyond
the classic optima.

One could argue that the apparent learning we observe is not learning at all
but a selection effect. By this logic, a common cause (e.g., higher
intelligence) is responsible both for players persisting longer at the game and
winning more often.  To check this, we created Figure~\ref{fig:prob_win7}, in
the Appendix, which is a similar plot except it restricts to players who played
at least 7 games. Because we see very similar results
%in the first seven games,
with and without this restriction,
we conclude that Figure \ref{fig:prob_win} reflects learning and not selection
effects.

%% SS: This paragraph is new
Recall that our experiment had a 2$\times$3 design with either 7 or 15
boxes and one of three possible underlying distributions of the box values.
Figure \ref{fig:prob_win} shows the average probability of our subjects
winning in the 7 and 15 box treatments aggregated over the three different
distributions of box values.  Figure \ref{fig:prob_win6} shows the
probability of people winning in each of the six treatments of our
experiment.  First, observe that the probability of winning, indicated by
the red lines, increases towards the maximal win rate in each of the six
treatments.
In all six treatments, most of the learning happens
in the early games with diminishing returns to playing more games. This
qualitative similarity demonstrates the robustness of the finding that there
is rapid and substantial learning in just a few repeated plays of the secretary
problem.
By comparing the probability of winning across all six treatments one can see that the low
and medium distributions were about equally as difficult and that the high
distribution was the hardest as it had the lowest probability of winning.  To
understand this, we next examine the types of errors the participants made. 
The probability of stopping after the maximum box value, indicated by the
green lines, is fairly similar across all six treatments.  There is,
however, variation in the probability of stopping before the max, indicated
by the blue lines, across the six treatments.  Participants were more
likely to stop before the maximum box in the high condition than the
others which explains why the subjects found this treatment to be the hardest.
Since the overall qualitative trends are fairly similar across the three
different distributions of box values we will aggregate over them in the
analyses that follow.

\begin{figure}
\centering
\includegraphics[width=0.5\textwidth]{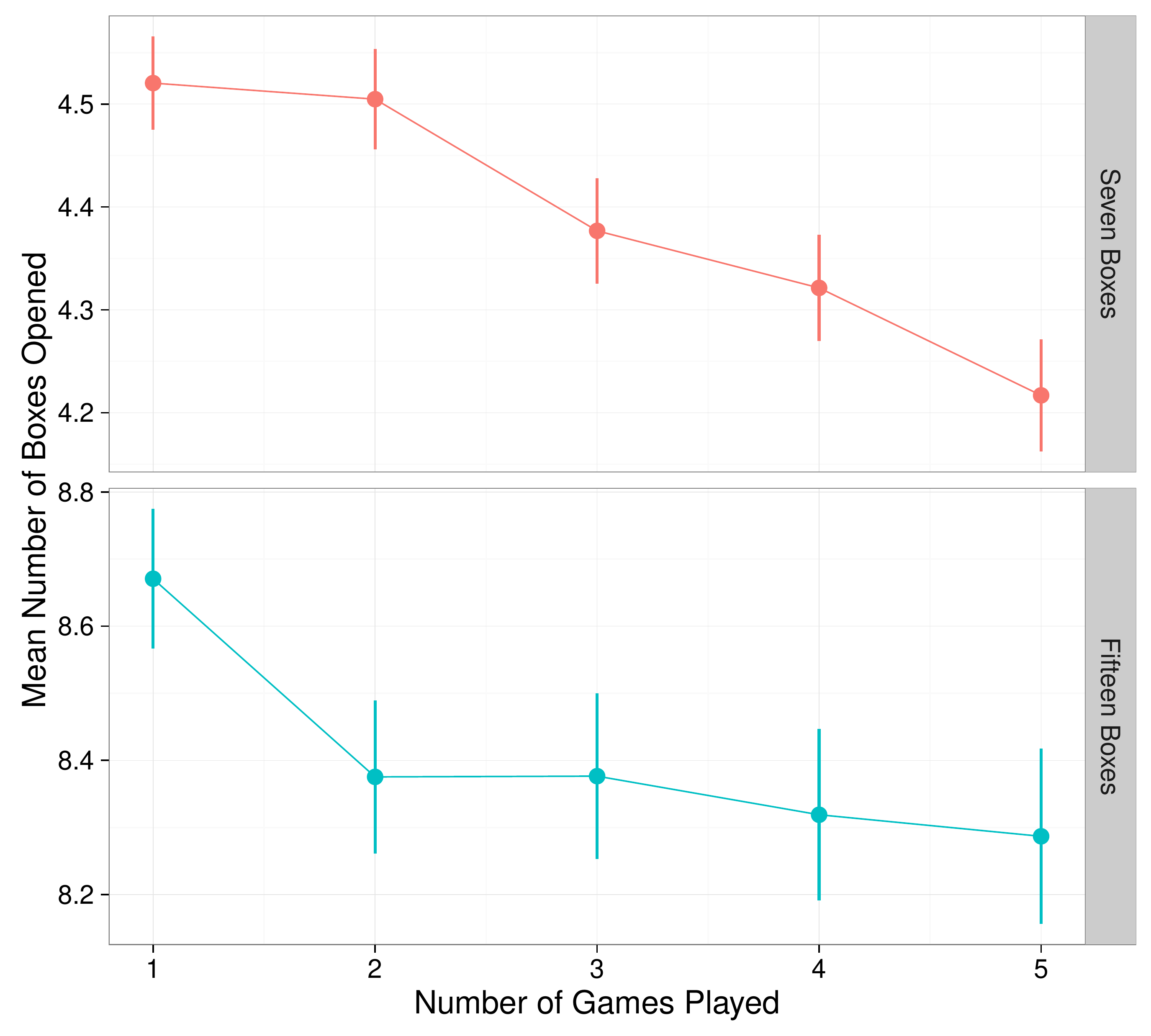}
%\vspace{-0.5\baselineskip}
\caption{Search depth for players in their first games measured by the number of boxes opened. Error
  bars indicate $\pm 1$ standard error}
%\vspace{-1\baselineskip}
\label{fig:search_depth}
\end{figure}

Having established that players do learn from experience, we turn our attention
to what is being learned. One overarching trend is that soon after their first
game, people learn to search less. As seen in Figure \ref{fig:search_depth}, in
the first five games, the depth of search decreases by about a third of one
box. Players can lose by stopping too early or too late. These search depth
results suggest that stopping too late is the primary concern that
participants address early in their sequence of games. This is also reflected
in the rate of decrease in the ``stopping after max'' errors in Figure
\ref{fig:prob_win} and \ref{fig:prob_win6}.  In both panels, rates of stopping after the maximum
decrease most rapidly.

\begin{figure*}
  \centering
  \subfigure[Human Players]{
\includegraphics[width=0.48\textwidth]{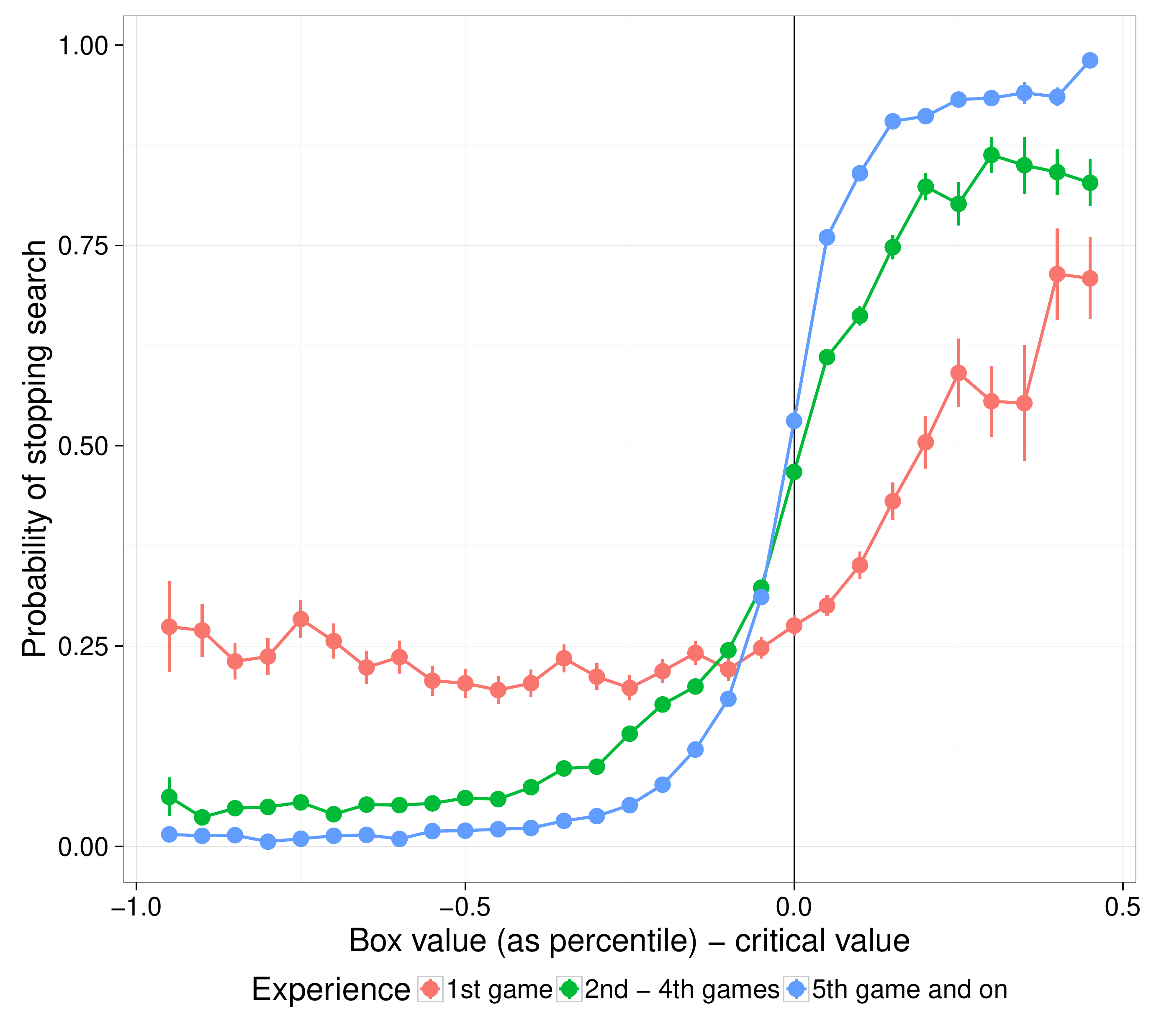}
\label{fig:prob_stop}
  }%
  \subfigure[LP Agent]{
\includegraphics[width=0.48\textwidth]{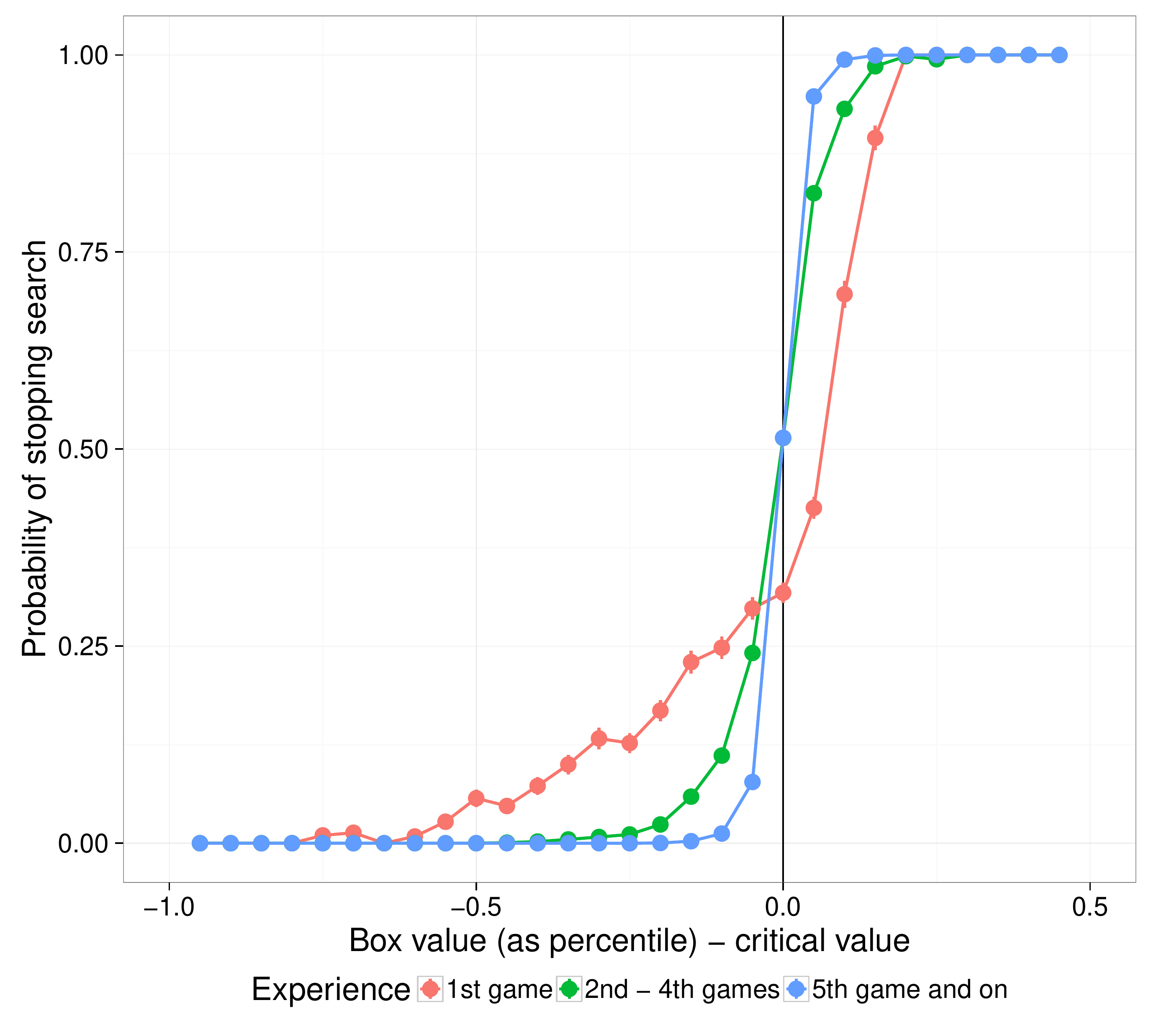}
\label{fig:prob_stop_sim}
  }
%\vspace{-0.5\baselineskip}
  \caption{
Left panel: Empirical rates of stopping search for box values above and
below the critical values. 
Right panel: Version of \ref{fig:prob_stop} with data from simulated agents instead of human players. 
Only non-dominated boxes are included in this analysis. 
}
%\vspace{-1\baselineskip}
\label{fig:prob_stop_reg_and_sim}
\end{figure*}

\subsection{Optimality of box-by-box decisions}
\label{sec:optimality}
Do players' decisions become
more optimal with experience? Recall that when the distribution is known
one can make an optimal decision about when to stop
search by comparing the percentile of an observed box value to the relevant critical
value in Table~\ref{tbl:1}. If the observed value exceeds the critical value,
it is optimal to stop, otherwise it is optimal to continue search. In Figure
\ref{fig:prob_stop_reg_and_sim}, the horizontal axis shows the difference
between observed box values (as percentiles) and the critical values given in
Table~\ref{tbl:1}. The vertical axis shows the probability of stopping search
when values above or below the critical values are encountered. The data in the
left panel are from human players and reflect all box-by-box decisions. 

An optimal player who knows the exact percentile of any box value, as well as
the critical values, would always 
keep searching (stop with probability 0) when encountering a value whose percentile is below the critical value. 
Similarly, such an optimal player would always stop searching (stop with probability 1) when encountering
a value whose percentile exceeds the critical value.
Together these two behaviors would lead to a step function: stopping with probability 0 to the left of the
critical value and stopping with probability 1 above it. 

Figure \ref{fig:prob_stop} shows that on first games (in red), players tend to
both under-search (stopping about 25\% of the time when below the critical
value) and to over-search (stopping at a maximum of 75\% of the time instead of
100\% of the time when above the critical value).  In a player's second through
fourth games (in green) performance is much improved, and the probability of
stopping search is close to the ideal .5 at the critical value. The blue curve,
showing performance in later games, approaches ideal step function. To address
possible selection effects in this analysis, Figure \ref{fig:prob_stop_select}
in the Appendix is similar to Figure \ref{fig:prob_stop_reg_and_sim} except it
restricts to the games of those who played a substantial number of games.
Because there are fewer observations, the error bars are larger but the overall
trends are the same suggesting again that these results are due to learning as
opposed to selection bias.

Attaining ideal step-function performance is not realistic when learning about
the distribution from experience. Comparison to the LP agent provides a
baseline of how well one could ever hope to do. Figure~\ref{fig:prob_stop_sim}
shows that in early games, even the LP agent both stops and continues when it
should not. Failing to obey the optimal critical values is a necessary
consequence of learning about a distribution from experience. Compared to the
human players, however, the LP agent approaches optimality more rapidly.
Furthermore, on the first game, it is less likely to make large-magnitude
errors. While the human players never reach the ideal stopping rates of 0 and 1
on the first game, the LP agent does so when the observed values are
sufficiently far from the critical vales.

Figure~\ref{fig:prob_stop} shows that stopping decisions stay surprisingly
close to optimal thresholds in aggregate.  Recall that the optimal thresholds
depend on how many boxes are left to be opened (see Table~\ref{tbl:1}). Because 
early boxes are encountered more often than late ones, this analysis could
be dominated by decisions on the early boxes. 
To address this, in what follows we estimate the threshold of each
box individually.

\subsection{Effects of unhelpful feedback} 

One may view winning or losing the
game as a type of feedback for the player to indicate if the strategy used
needs adjusting.  Taking this view, consider a player's first game. Say this
player over-searched in the first game, that is, they saw a value greater than the
critical value but did not stop on it. Assume further that this player won this
game.  This player did not play the optimal strategy but won anyway, so their
feedback was unhelpful. The middle panel of Figure~\ref{fig:unhelpful_first}
shows the errors made during a second game after over-searching and either
winning or losing during their first game. 
The red curve tends to be above the blue curve, meaning that players who stopped
too late but didn't get punished (blue) are less likely to stop on most box values in the next game, compared
to players who stopped too late and got punished (red).

Similarly, the bottom panel shows the blue curve to be above the red curve, meaning 
that players who stopped too early but didn't get punished (blue) are more likely to stop on most box values 
in the next game, compared to players who stopped too early and got punished (red).

This finding makes the results in Figures \ref{fig:prob_win} and
\ref{fig:prob_stop_reg_and_sim} even more striking as it is a reminder that the
participants are learning in an environment where the feedback they receive is
noisy.  Figure~\ref{fig:unhelpful_second} shows the errors in the fifth game
given the feedback from the first game.  Even a quick glance shows that the
curves are essentially on top of each other.  Thus, those who received
unhelpful feedback in the first game were able to recover---and perform just as
well as those who received helpful feedback---by the fifth game.

\begin{figure}[t]
\centering
  \subfigure[]{
\includegraphics[width=0.47\textwidth]{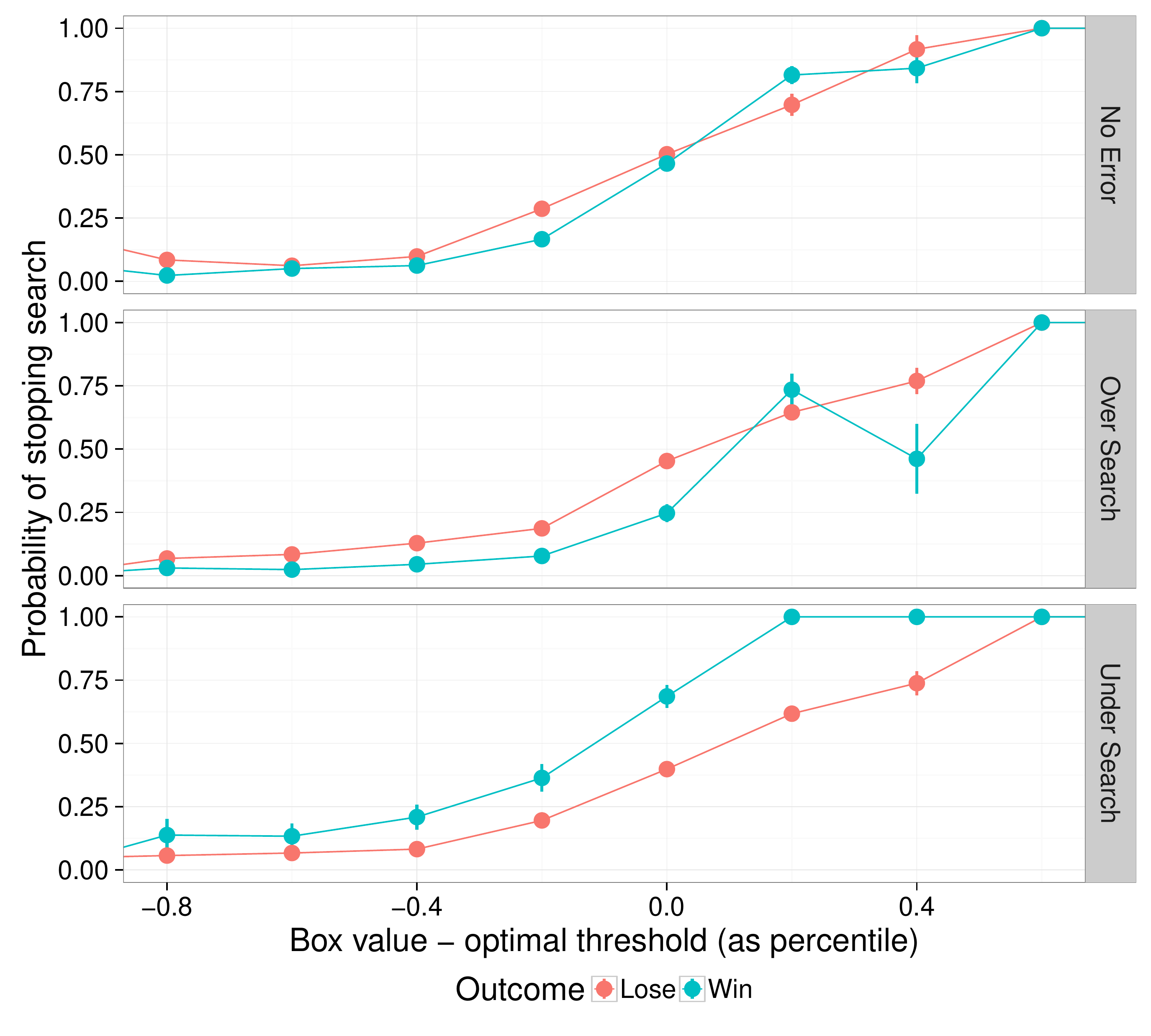}
\label{fig:unhelpful_first}
  }
  \subfigure[]{
\includegraphics[width=0.47\textwidth]{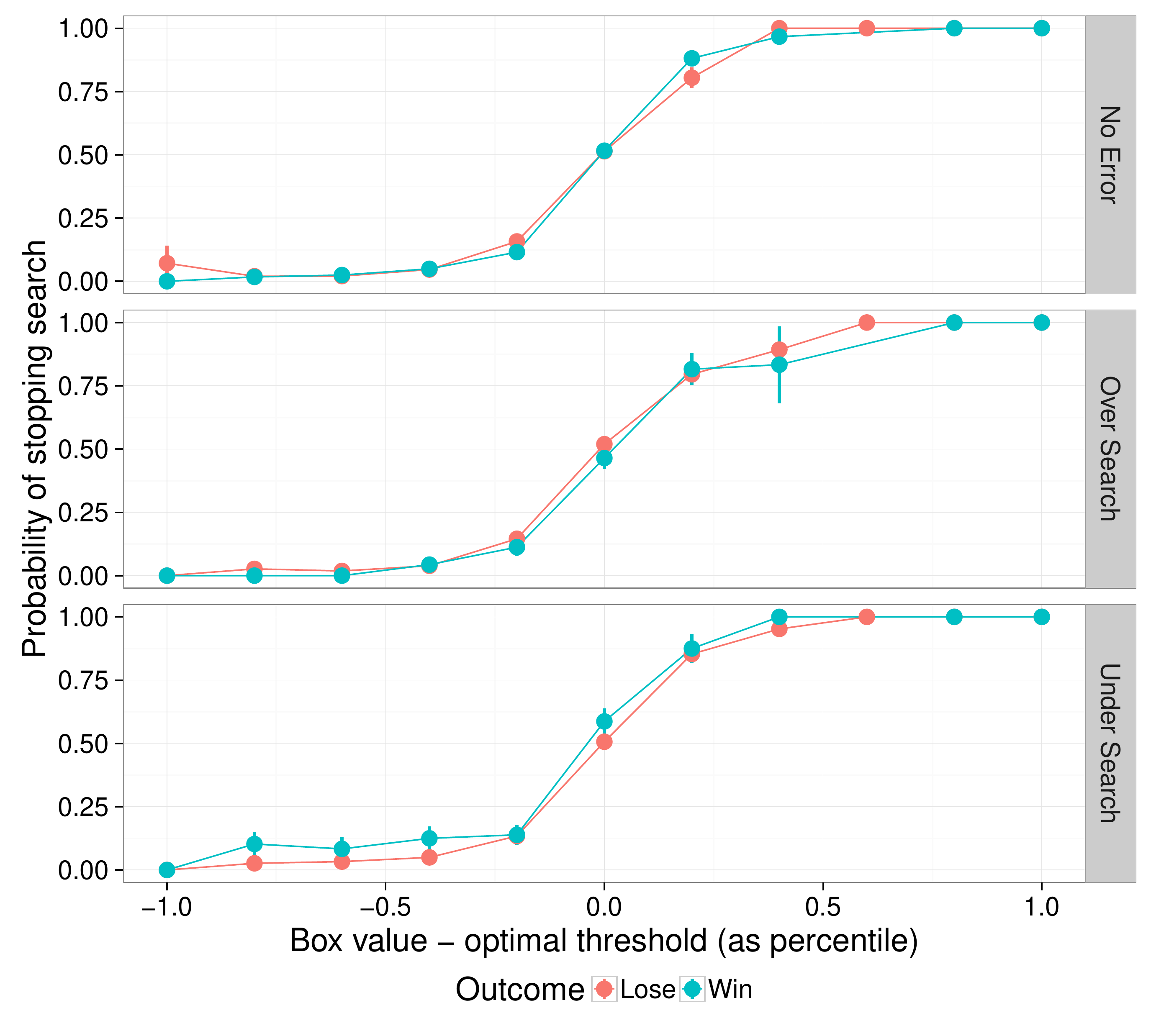}
\label{fig:unhelpful_second}
  }
%\vspace{-0.5\baselineskip}
\caption{Errors in the second (left) and fifth (right) games given 
    whether the first game was won (red curves) or lost (blue curves).
    Vertically arranged panels indicate what type of error, if any, was
    made on the first game. }
%\vspace{-0.75\baselineskip}
\label{fig:unhelpful_feedback}
\end{figure}

%% file: models_header.tex
% -*- mode:LaTeX; mode:flyspell; -*-
\section{Modeling Player Decisions}
\label{sec:model}

%% Overview: comparing models
In this section we explore the predictive performance of several models of human behavior in the repeated secretary problem with learnable distributions.
We begin by describing our framework for evaluating predictive models, then describe the models, and finally compare their performance.

\subsection{Evaluation and comparison}
\label{sec:evaluation}

%% Probabilistic modeling
Our goal in this section is to compare several psychologically plausible models in terms of how well they capture human behavior in the repeated secretary game to give us some insight as to how people are learning to play the game.
Since our goal is to compare how likely each model is given the data the humans generated we use a Bayesian model comparison framework.
The models we compare, defined in Section~\ref{sec:model-defns}, are probabilistic, allowing them to express differing degrees of confidence in any given prediction.  This also allows them to capture heterogeneity between players. In contexts where players' actions are relatively homogeneous, their actions can be predicted with a high degree of confidence, whereas in contexts where players' actions differ, the model can assign probability to each action.

%% Stopping decisions and histories
After opening each box, a player makes a binary decision about whether or not to stop.
Our dataset consists of a set of \emph{stopping decisions} $y^g_{t} \in \{0,1\}$ that the player made in game $g$ after seeing non-dominated box $t$.  If the player stopped at box $t$ in game $g$, then $y^g_{t}=1$; otherwise, $y^g_{t}=0$.
Our dataset also contains the \emph{history} $x^g_{T:t} = \left(x^g_T, x^g_{T-1}, \ldots, x^g_t\right)$ of box values that the player had seen until each stopping decision.  We represent the full dataset by the notation $\D$.

%% Probability model
In our setting, a probabilistic model $f$ maps from a history $x^g_{T:t}$ to a probability that the agent will stop.
(This fully characterizes the agent's binary stopping decision.)
Each model may take a vector $\theta$ of parameters as input.
We assume that every decision is independent of the others, given the context.
Hence, given a model and a vector of parameters, the likelihood of our dataset is the product of the probabilities of its decisions; that is,
\[p\left(\D \given h, \theta \right) =
\prod_{(x^g_{T:t},y^g_t) \in \D} \left[ f\left(x^g_{T:t} \given \theta\right)y^g_t +
\left(1-f\left(x^g_{T:t} \given \theta\right)\right)\left(1-x^g_{T:t}\right) \right].
\]

%% Bayesian model comparison
In Bayesian model comparison, models are compared by how probable they are given the data.
That is, a model $f^1$ is said to have better predictive performance than model $f^2$ if
$p\left(f^1 \given \D\right) > p\left(f^2 \given \D\right)$, where 
\begin{equation}
    p\left(f \given \D\right) = \frac{p(f)p\left(\D \given f\right)}{p(\D)}. \label{eqn:model-probs}
\end{equation}
With no a priori reason to prefer any specific model, we can assign them equal prior model probabilities $p(f)$.
Comparing the model probabilities defined in Equation~\eqref{eqn:model-probs} is thus equivalent to comparing the models' \emph{model evidence}, defined as
\begin{equation}
    p\left(\D \given f\right) = \int_\Theta p\left(\D \given f, \theta \right)p(\theta)d\theta. \label{eqn:evidence}
\end{equation}
The ratio of model evidences $p\left(\D \given f^1\right) \big/ p\left(\D \given f^2\right)$ is called the \emph{Bayes factor} \citep[e.g., see][]{kruschke-doing}.  The larger the Bayes factor, the stronger the evidence in favor of $f^1$ versus $f^2$.

%% Advantages of the Bayes factor
This probabilistic approach has several advantages.
First, the Bayes factor between two models has a direct interpretation: it is the ratio of probabilities of one model's being the true generating model, conditional on one of the models under consideration being the true model.
Second, it allows models to quantify the confidence of their predictions.
This quantification allows us to distinguish between models that are almost correct and those that are far from correct in a way that is impossible for coarser-grained comparisons such as predictive accuracy.

%% Overfitting / Bayes factor
One additional advantage of the Bayes factor is that it provides some compensation for overfitting.
Models with a higher dimensional parameter space are penalized, due to the fact that the integral in Equation~\eqref{eqn:evidence} must average over a larger space.  The more flexible the model, the more of this space will have low likelihood, and hence the better the fit must be in the high-probability regions in order to attain the same evidence as a lower-parameter model.
%In particular, this means that when one model generalizes another but has  equivalent (or even insufficiently better) fit at its best-fitting parameters, the more restricted model will have a high Bayes factor relative to the generalized model.

%% Cross validation
The amount by which high-dimensional models are penalized by the Bayes factor depends strongly upon the choice of prior.  This dependence upon the relatively arbitrary choice of prior is undesirable.
An alternative approach is to evaluate models using \emph{cross-validation}, in which the data are split into a \emph{training set} that is used to set the parameters of the model, and a \emph{test set} that is used to evaluate the model's performance.

%% Hybrid: Cross-validated Bayes factor
We use a hybrid of the cross-validation and Bayesian approaches.
We first randomly select a split $s=(\D^{\textup{train}}, \D^{\textup{test}})$, with
$\D^{\textup{train}} \union \D^{\textup{test}} = \D$ and $\D^{\textup{train}} \intersect \D^{\textup{train}} = \varnothing$.
We then compute the \emph{cross-validated model evidence} of the test set $\D^{\textup{test}}$ with respect to the prior updated by the training set $p\left(\theta \given \D^{\textup{train}}\right)$, rather than computing the model evidence of the full dataset $\D$ with respect to the prior $p(\theta)$.
To reduce variance due to the randomness introduced by the random split, we take expectation over the split, yielding
\begin{equation}
    \E_s p\left(\D^{\textup{test}} \given f, \D^{\textup{train}}\right) = \int\left[\int_\Theta p\left(\D^{\textup{test}} \given f, \theta \right)p\left(\theta \given \D^{\textup{train}}\right)d\theta\right]p(s)ds, \label{eqn:cv-evidence}
\end{equation}
where $p(s)$ is a uniform distribution over all splits.
The ratio of cross-validated model evidences
\begin{equation}
   \frac
   {\E_s p\left(\D^{\textup{test}} \given f^1, \D^{\textup{train}}\right)}
   {\E_s p\left(\D^{\textup{test}} \given f^2, \D^{\textup{train}}\right)}\label{eqn:cv-bayes}
\end{equation}
is called the \emph{cross-validated Bayes factor}.  As with the Bayes factor, larger values of \eqref{eqn:cv-bayes} indicate stronger evidence in favor of $f^1$ versus $f^2$.

%% Estimation
The integral in Equation~\eqref{eqn:cv-evidence} is analytically intractable, so we followed the standard practice of approximating it using Markov chain Monte Carlo sampling.  Specifically, we used the PyMC software package's implementation \citep{salvatier-probabilistic} of the Slice sampler \citep{neal-slice} to generate $25000$ samples from each posterior distribution of interest, discarding the first $5000$ as a ``burn in'' period.
We then used the ``bronze estimator'' of \citet{alqallaf-crossvalidation} to estimate Equation~\eqref{eqn:cv-evidence}
based on this posterior sample.
% We then used the Gelfand-Dey method \citep{gelfand-dey} to estimate Equation~\eqref{eqn:evidence}
% based on this posterior sample.\footnote{This is nontrivial because the integral is with respect to the prior, not the posterior. However, most of the contribution to the integral's total comes from high-posterior regions of the parameter space, so simply sampling from the prior would produce a very noisy estimate.}

%  LocalWords:  learnable

%% file: models.tex
% -*- mode:LaTeX; mode:flyspell; -*-

\subsection{Models}
\label{sec:model-defns}

%% Model framework
We start by defining our candidate models,
each of which assumes that an agent decides at each non-dominated box whether
to stop or continue, based on the history of play until that point.
For notational convenience, we represent a history of play by a tuple containing the number of boxes seen $i$, the number of non-dominated boxes seen $i^*$, and the percentile of the the current box $q_i$ as estimated using Equation~\ref{eqn:pertile}.
% rather than as a list of all box values seen to date.
% This is without loss of generality, since each of the models we consider (including the optimal model) makes this decision based on some combination of these three values.
% \sid{I'd like to add a sentence here that says something to the effect that there is no other data that a model could use to decide to stop or not.}
% \jrw{I've reworded to express this as being without loss of generality, to explicitly call out the optimal model, and to explicitly say that it's about notational convenience.}
Formally, each model is a function $f:\N \times \N \times [0,1] \to [0,1]$ that maps from a tuple $(i, i^*, q_i)$ to a probability of stopping at the current box.

\begin{definition}[Value Oblivious]
    In the Value Oblivious model, agents do not attend to the specific box values.
    Instead, conditional upon reaching a non-dominated box $i$, an agent stops with a fixed probability $p_i$.
    \[f^\textup{value-oblivious}\left(i, i^*, q_i \,\middle|\, \left\{p_j\right\}_{j=1}^{T-1}\right) = p_i.\]
\end{definition}

%% Viable k
\begin{definition}[Viable~k]
    The Viable~k model stops on the $k$th non-dominated box.
    \[f^\textup{viablek}\left(i, i^*, q_i \,\middle|\, k, \epsilon\right) =
    \begin{cases}
       \epsilon   &\text{if } i^* < k,\\
       1-\epsilon &\text{otherwise.}
   \end{cases}
   \]
   In this model and the next agents are assumed to err with probability $\epsilon$ on any given decision.
\end{definition}

%% Sample k
\begin{definition}[Sample~k]
    The Sample~k model stops on the first non-dominated box that it encounters after having seen at least $k$ boxes, whether those boxes were dominated or not.  
    \[f^\textup{sample}\left(i, i^*, q_i \,\middle|\, k, \epsilon\right) =
    \begin{cases}
       \epsilon   &\text{if } i < k,\\
       1-\epsilon &\text{otherwise.}
   \end{cases}
   \]
   When $k = \lceil{T/e}\rceil$ and $\epsilon=0$, this corresponds to the optimal solution of the classical secretary problem in which the distribution is unknown.
\end{definition}

%% Threshold (est)
\begin{definition}[Multiple Threshold]
The Multiple Threshold model stops at box $i$ with increasing probability as the box value increases.
We use a logistic specification which yields a sigmoid function at each box $i$
such that at values equal to the threshold $\tau_i$ an agent stops with probability
0.5; an agent stops with greater (less) than 0.5 probability on values higher (lower) than $\tau_i$, with the probabilities becoming more certain as the value's distance from $\tau_i$ grows.
%Naturally, if box $i$ has value higher (or lower) than $\tau_i$ it will be stopped on with probability greater than (or less than) 0.5.
%\sid{Is the previous sentence too much hand holding?}
%\jrw{I used a slightly different wording to make it even more explicit}
% We use a logistic specification, learning a box-specific threshold $\tau_i$
% for each box that represents the value at which an agent stops with
% probability 0.5.  
We also learn a single parameter $\lambda$ across all boxes representing how
quickly the probability changes as a box value becomes further from
$\tau_i$. Intuitively, $\lambda$ controls the slope of the
sigmoid\footnote{We considered models with one $\lambda$ per box but
  they did not perform appreciably better than the single $\lambda$ models.}.
    \[f^\textup{thresholds}\left(i, i^*, q_i \,\middle|\, \lambda, \left\{\tau_j\right\}_{j=1}^{T-1}\right) =
    \frac{1}{1 + \exp[\lambda(q_i - \tau_i)]}.
    \]
When the thresholds are set to the critical values of Table~\ref{tbl:1} such that $\tau_i = z_{(T-i+1)}$, this model corresponds to the optimal solution of the secretary problem with a known distribution.
\end{definition}

%% Single threshold (est)
\begin{definition}[Single Threshold]
The Single Threshold model is a simplified threshold model in which agents compare all box values to a single threshold $\tau$ rather than box-specific thresholds.    
    \[f^\textup{single-threshold}\left(i, i^*, q_i \,\middle|\, \lambda, \tau\right) =
    \frac{1}{1 + \exp[\lambda(q_i - \tau)]}.
    \]
\end{definition}

%% Two threshold (est)
\begin{definition}[Two Threshold]
\label{def:two-thresh}
The Two Threshold model is another simplified threshold model in which agents compare all ``early'' box values (the first $\lfloor{T/2}\rfloor$ boxes) to one threshold ($\tau_0$) and all ``late'' box values to another ($\tau_1$). 
    \[f^\textup{two-threshold}\left(i, i^*, q_i \,\middle|\, \lambda, \tau_0, \tau_1\right) =
    \frac{1}{1 + \exp[\lambda(q_i - \tau_\nu)]},
    \]
where
\[\tau_\nu = \begin{cases}
                \tau_0 &\text{if } i < \frac{T}{2}, \\
                \tau_1 &\text{otherwise.}
            \end{cases}
\]
This is a low-parameter specification that nevertheless allows for differing behavior as a game progresses (unlike Single Threshold).
\end{definition}

\paragraph{Priors.}
Each of the models described above has free parameters that must be estimated from the data.
We used the following uninformative prior distributions for each parameter:
\begin{align*}
    p_i          &\sim \text{Uniform}[0, 1]                   &\tau, \tau_i, \tau_0, \tau_1 &\sim \text{Uniform}[0, 1] \\
    k            &\sim \text{Uniform}\{1, 2, \ldots, T-1\}    &\lambda      &\sim \text{Exponential}(\mu=1000).\\
    \epsilon     &\sim \text{Uniform}[0,0.5]
\end{align*}
The hyperparameter $\mu$ for precision parameters $\lambda$ was chosen manually to ensure good mixing of the sampler.
Each parameter's prior is independent; e.g., in the single-threshold model a given pair $(\lambda, \tau)$ has prior probability $p(\lambda, \tau) = p(\lambda)p(\tau)$.

\subsection{Model comparison results}
\label{sec:comparison-results}
\begin{figure*}[tb]
\centering
  \subfigure[Seven Boxes]{
\includegraphics[width=0.48\textwidth]{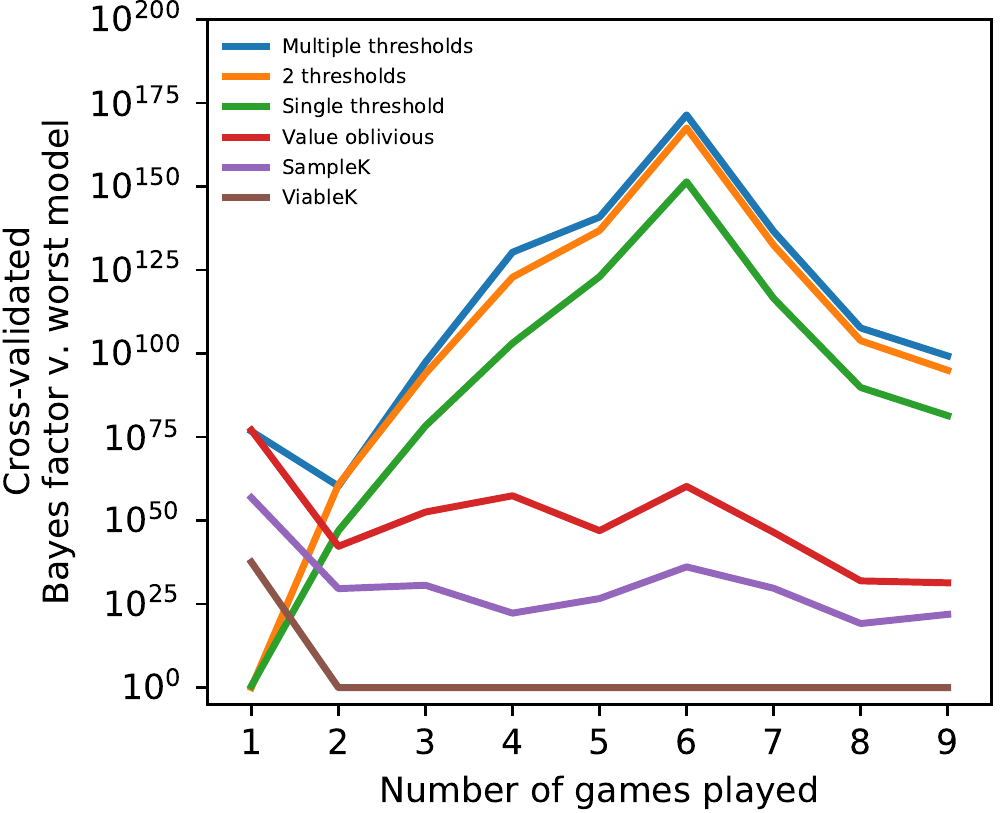}
\label{fig:bayes-factors7}
  }%
  \subfigure[Fifteen Boxes]{
\includegraphics[width=0.48\textwidth]{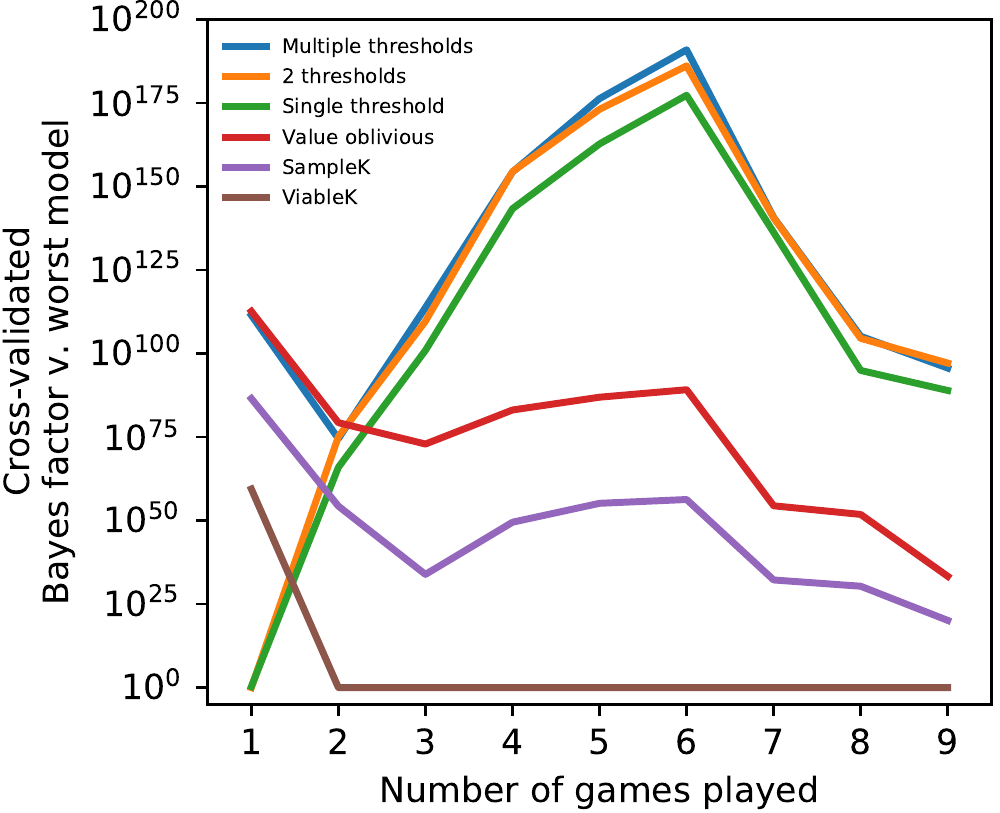}
\label{fig:bayes-factors15}
  }
%\vspace{-1\baselineskip}
  \caption{Cross-validated Bayes factors for various models, compared to the lowest-evidence model in each game.}
%\vspace{-1\baselineskip}
\label{fig:bayes-factors}
\end{figure*}

%% fig:bayes-factors description
Figure~\ref{fig:bayes-factors} gives the cross-validated Bayes factors for each of the models of Section~\ref{sec:model-defns}.
%% Games estimated separately
The models were estimated separately for each number of games; that is,
each model was estimated once on all the first games played by
participants, again on all the second games, etc.
This allows us to detect learning by comparing the estimated values of the parameters across games.
%\sid{James, can you add
% one sentence justifying why estimating the models separately for each
% game played is the right thing to do? Also, in the next sentence is that
% a standard thing to do?}
%% Bayes factors normalized against worst model
The cross-validated Bayes factor is defined as a ratio between two cross-validated model evidences.
Since we are instead comparing multiple models, we take the standard approach of expressing each factor with respect to the lowest-evidence model for a given number of games.
These normalized cross-validated Bayes factors are consistent, in the sense that if the normalized cross-validated Bayes factor for model $h^1$ is $k$ times larger than the normalized cross-validated Bayes factor for $h^2$, then the cross-validated Bayes factor between $h^1$ and $h^2$ is $k$.
%{worst_evidence}
As a concrete example, the Two~Threshold model had the lowest cross-validated model evidence for participants' first games in the seven box condition; the cross-validated model evidence for the Value Oblivious model was $10^{21}$ times greater than that of the Sample~k model, and $10^{40}$ times greater than that of the Viable~k model.

%% Best-performing models
%{model_comparisons}
In first game played, in both the seven box and fifteen box conditions, and in second game in the fifteen box condition, the best performing model was Value Oblivious.
In all subsequent games and in both conditions, Viable~k was the worst performing model and the Multiple~Threshold model was the best performing model.
%
%% Two Threshold next best
The Two~Threshold model was the next-best performing model.\footnote{We tested whether a different boundary between the two thresholds would perform better by estimating a model in which the choice of boundary was a separate parameter.
This model actually performed worse that the Two~Threshold model.
This indicates that the data do not argue strongly for a different boundary, and hence the only effect of adding a boundary parameter was overfitting.}
%
%%% Two Threshold much closer in 15 box condition than 7 box condition
% On the 15 box condition, the gap between the Two~Threshold and the Multiple~Threshold models was much smaller than on the seven box condition.
%
%%% Two Threshold model does ``almost as well''?
% It may appear that the Two~Threshold model performed ``almost as well'' as the Multiple~Threshold model.  However, this is a scale artifact due to the very large gap between the Multiple~Threshold model and the other models in Figure~\ref{fig:bayes-factors}.
% %{two vs multiple}
% The cross-validated Bayes factor in favor of the Multiple~Threshold model versus the Two~Threshold model ranged between approximately $8$ \citep[``substantial'' according to][]{jeffreys-probability} and $10^{8}$ \citep[``decisive'' per][]{jeffreys-probability}.

%% Consistent with improving outcomes
Evidently, players behaved consistently with the optimal class of model for the known distribution---multiple thresholds---as early as the second game.
This is consistent with the observations of Section~\ref{sec:optimality}, in which players' outcomes improved with repeated play.
In addition, it is consistent with the learning of optimal thresholds in Figure~\ref{fig:prob_stop} but improves on that analysis because here the most common stopping points---the early boxes---do not dominate the average.
\begin{figure}
    \centering
    \includegraphics[width=0.98\textwidth]{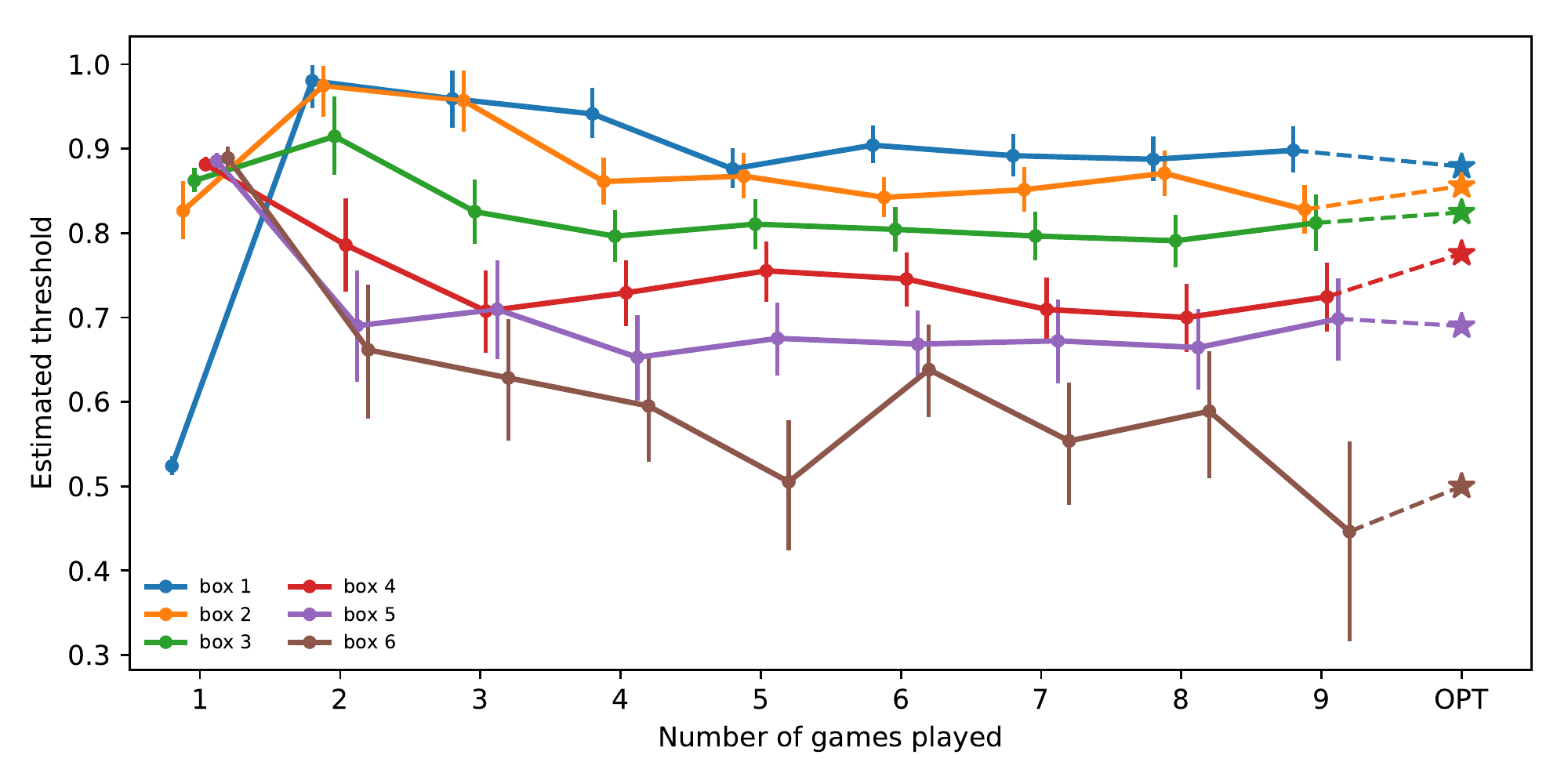}
    %\vspace{-1\baselineskip}
    \caption{Estimated thresholds in the seven box games.
    The rightmost set of points are the optimal thresholds.
    Error bars represent the 95\% posterior credible interval.}
    %\vspace{-1\baselineskip}
    \label{fig:thresholds7}    
\end{figure}
% \dgg{How about making Fig 8 bigger? It's kind of the money shot}.
% \jrw{Sure thing; it was kind of an accident that it was so small.  I can change the aspect ratio and regenerate if we need it to be shorter and wider.
% }
%% Threshold convergence
Furthermore, players' estimated thresholds approached the theoretically optimal values remarkably quickly.
Figure~\ref{fig:thresholds7} shows the estimated thresholds for the seven box condition, along with their 95\% posterior credible intervals.
The estimated thresholds for the second and subsequent games are strictly decreasing in the number of boxes seen, like the optimal thresholds.
Overall, the thresholds appear to more closely approximate their optimal values over time.
After only four games, each threshold's credible interval contains the
optimal threshold value.\footnote{In games 5--8, either one or two credible
  intervals no longer contain the corresponding optimal value; by game 9
  all thresholds' credible intervals again contain their optimal values.}
Thus, workers learned to play according to the optimal family of models and
learned the optimal threshold settings within that family of models.

%% First game
The success of the Value Oblivious model in the first game (Figure~\ref{fig:bayes-factors}) suggests that neither of the threshold-based models fully capture players' decision making in their initial game.  This is further supported by the best-estimates of thresholds for the first game: unlike subsequent games which have thresholds that strictly decrease in number of boxes seen, in the first game the estimated thresholds are strictly \emph{increasing} in number of boxes seen.
This is consistent with players using a Value Oblivious model. If players who stop on later boxes do so for reasons independent of the box's value, then they will tend to stop on higher values merely due to the selection effect from only stopping on non-dominated boxes.

In sum, the switch from increasing to decreasing thresholds in Figure~\ref{fig:thresholds7} is consistent with moving from a value-oblivious strategy, which generalizes the optimal solution for the classical problem, to a threshold strategy, which generalizes the optimal strategy for known distributions. 
%%However, the non-threshold Sample~k and Viable~k models also turn out to be comparatively poor predictors of stopping behavior.
%%Not super relevant and makes for a poor transition.
   
%%  LocalWords:  learnable

%% file: conclusion.tex
\section{Conclusion: Behavioral Insights}
\label{sec:conclusion}
The main research question we
addressed in this work is whether people improve at the secretary problem
through repeated play. In contrast to prior research
\citep{campbell-secretary, lee-secretary, seale-secretary}, across thousands
of players and tens of thousands of games, we document fast and steep learning
effects. Rates of winning increase by about 25 percentage points over the
course of the first ten games (Figure~\ref{fig:prob_win}).

From the results in this article, it seems as if players not only improve, but also 
learn to play in a way that approaches optimality in several
respects, which we list here. Rates of winning come within about five
percentage points of the maximum win rates possible, and this average is taken
without cleaning the data of players who were obviously not trying. In looking
at box-by-box decision making, player's probabilities of stopping came to
approximate an optimal step function after a handful of games
(Figure~\ref{fig:prob_stop_reg_and_sim}). And similar deviations from the
optimal pattern were also observed in a very idealized agent that learns from
data, suggesting that some initial deviation from optimality is
inevitable. 
%% Not
%% only did players' decisions converge to the optimal thresholds in general, they
%% seemed to converge to each of the optimal thresholds independently
%% (Figure~\ref{fig:regression7}).  
Perhaps even more remarkably, they were able
to do this with no prior knowledge of the distribution and, consequentially,
sometimes unhelpful feedback (Figure~\ref{fig:unhelpful_feedback}).

In the first game, player behavior was relatively well fit by the Value
Oblivious model which had a fixed probability of stopping at each box,
independent of the values of the boxes. In later plays,
threshold-based decision making---the optimal strategy for known
distributions---fit the data best (Figure~\ref{fig:bayes-factors}).
Further analyses uncovered that players' implicit thresholds were close to the
optimal critical values (Figure~\ref{fig:thresholds7}), which is surprising given the small
likelihood that players actually would, or could, solve for these values. 

A few points of difference could explain the apparent departure from prior
empirical results. First, to our knowledge, ours is the first study to begin
with an unknown distribution that players can learn over time. Seemingly small
differences in instructions to participants could have a large effect.  As
mentioned, other studies have informed participants about the distribution, for
example its minimum, maximum, and shape. Second, some prior experimental
designs have presented ranks or manipulated values that made it difficult to
impossible for participants to learn the distributions. Third, past studies
have used relatively few participants, making it difficult to detect learning
effects. For example, \citet{campbell-secretary} have 12 to 14 participants per
condition and assess learning by binning the first 40, second 40, and third 40
games played. In contrast, with over 5,000 participants, we can examine success
rates at every number of games played beneath 20 with large sample sizes.
This turns out to be important for testing learning, as most of it happens in
the first 10 games. While our setting is different than prior ones, the change
of focus seems merited because many real-world search problems (such as hiring
employees in a city) involve repeated searches from learnable distributions.

A promising direction for future research would be to propose and test a
unified model of search behavior that can capture several properties observed
here such as: 
the effects on unhelpful feedback (Figure~\ref{fig:unhelpful_feedback}), 
the transition from value-oblivious to threshold-based decision making (Figure~\ref{fig:bayes-factors}), 
and the learning of near-optimal
thresholds (Figure~\ref{fig:thresholds7}). Having established that people learn
to approximate optimal stopping in repeated searches through
distributions of candidates, the next challenge is to model how individual 
strategies evolve with experience.

%% file: appendix.tex
\section{Computation of critical values and probability of winning in known distribution case}
\label{sec:comp_cv_pwin}

Note that
\begin{equation} \label{GrindEQ__1_} 
p_{1} (h)=1-F(h). 
\end{equation} 
This is the probability of observing something on the last round that exceeds the best observation.  Generally

\begin{equation}
\begin{aligned}
 p_{t}(h) &=\left\{\begin{array}{cc} {\int _{h}^{\infty }F(x)^{t-1} f(x)dx +F(h)p_{t-1} (h)} & {c_{t} \le h} \\ {\int _{c_{t} }^{\infty }F(x)^{t-1} f(x)dx +\int _{h}^{c_{t} }p_{t-1} (x)f(x)dx +F(h)p_{t-1} (h)} & {c_{t} >h} \end{array}\right.  \\
&=\left\{\begin{array}{cc} {\frac{1-F(h)^{t} }{t} +F(h)p_{t-1}
  (h)} & {c_{t} \le h} \\ {\frac{1-F(c_{t} )^{t} }{t} +\int _{h}^{c_{t}
  }p_{t-1} (x)f(x)dx +F(h)p_{t-1} (h)} & {c_{t} >h} \end{array}\right. 
\end{aligned}
\label{GrindEQ__2_} 
\end{equation}

To understand \eqref{GrindEQ__2_}, first note that if the critical value is
less than the historically best observation, anything exceeding the
historically best observation $h$ is acceptable.  Thus, if something better,
$x$, is observed, it is accepted, in which case the player wins if all the
subsequent observations are worse, with probability $F(x)^{t-1}$.  Otherwise,
we inherit $h$ and have a probability of winning $p_{t-1}(h)$ in the next
period.  

If $c_{t} > h$, then an acceptance occurs only if the realization $x$ exceeds
$c_{t}$, in which case the probability of winning remains $F(x)^{t-1}$.  If the
player experiences a value between $h$ and $c_{t}$, the historical maximum
rises but is not accepted.  Finally, if the observation is less than $h$, the
historically best observation is not incremented and the player moves to period
$t - 1$.

Let $\bar{p}_{t} $ be the value of $p_{t}$ arising when $c_{t} = h_{t}$.  Then

\begin{lemma}
$\bar{p}_{t} (h)=\sum _{j=1}^{t}\frac{F(h)^{j-1} -F(h)^{t} }{t+1-j}  $
\end{lemma}

\begin{proof}
The proof is by induction on $t$. The lemma is trivially satisfied at $t = 1$.  Suppose it is satisfied at $t - 1$.  Then, from \eqref{GrindEQ__2_},
%% \[\bar{p}_{t} (h)=\frac{1-F(h)^{t} }{t} +F(h)p_{t-1} (h)=\frac{1-F(h)^{t} }{t} +F(h)\sum _{j=1}^{t-1}\frac{F(h)^{j-1} -F(h)^{t-1} }{t-j}  \] 
%% \[=\frac{1-F(h)^{t} }{t} +\sum _{j=1}^{t-1}\frac{F(h)^{j} -F(h)^{t} }{t-j}  \] 
%% \[=\frac{1-F(h)^{t} }{t} +\sum _{j=2}^{t}\frac{F(h)^{j-1} -F(h)^{t} }{t+1-j}  =\sum _{j=1}^{t}\frac{F(h)^{j-1} -F(h)^{t} }{t+1-j}  \] 
\begin{flalign*}
\bar{p}_{t} (h) &=\frac{1-F(h)^{t} }{t} +F(h)p_{t-1} (h)=\frac{1-F(h)^{t} }{t} +F(h)\sum _{j=1}^{t-1}\frac{F(h)^{j-1} -F(h)^{t-1} }{t-j}  \\ 
&=\frac{1-F(h)^{t} }{t} +\sum _{j=1}^{t-1}\frac{F(h)^{j} -F(h)^{t} }{t-j}  \\ 
&=\frac{1-F(h)^{t} }{t} +\sum _{j=2}^{t}\frac{F(h)^{j-1} -F(h)^{t} }{t+1-j}  =\sum _{j=1}^{t}\frac{F(h)^{j-1} -F(h)^{t} }{t+1-j}   
\end{flalign*}
\end{proof}

\noindent Note that, at the value $x_{t} = c_{t}$, the searcher must be indifferent between accepting $x_{t}$ and rejecting, in which case the history becomes $c_{t}$.  Therefore,

\noindent 
\[\left. \frac{\partial p_{t} }{\partial c_{t} } \right|_{h_{t} =c_{t} } =0.\]

\noindent This gives
\[0=-f(c_{t} )F(c_{t} )^{t-1} +f(c_{t} )\bar{p}_{t-1} (c_{t} ),\]
or,
\begin{equation}
F(c_{t} )^{t-1} =\sum _{j=1}^{t-1}\frac{F(c_{t} )^{j-1} -F(c_{t} )^{t-1}
  }{t-j}.
\label{GrindEQ__3_} 
\end{equation}

\noindent Equation \ref{GrindEQ__3_} is intuitive, in that it says $F(c_{t} )^{t-1} =\bar{p}_{t-1} (c_{t} )$, that is, the probability of winning given an acceptance of $c_{t}$, which is $F(c_{t} )^{t-1} $, equals the probability of winning given that $c_{t}$ is rejected and becomes the going-forward history, which would give a probability of winning of $\bar{p}_{t-1} (c_{t} )$.  Thus,

\[p_{t} (h)=\left\{\begin{array}{cc} {\sum _{j=1}^{t}\frac{F(h)^{j-1} -F(h)^{t} }{t+1-j}  } & {c_{t} \le h} \\ {\frac{1-F(c_{t} )^{t} }{t} +\int _{h}^{c_{t} }p_{t-1} (x)f(x)dx +F(h)p_{t-1} (h)} & {c_{t} >h} \end{array}\right. \] 

\noindent Letting $i = T - t,$ and $z_{i} = F(c_{t})$ we can rewrite \eqref{GrindEQ__3_} to give

\noindent 
\begin{equation}
z_{t} ^{t-1} =\sum _{j=1}^{t-1}\frac{z_{t} ^{j-1} -z_{t} ^{t-1} }{t-j}
\textrm{,
or } z_{t} ^{t} =\sum _{j=1}^{t-1}\frac{z_{t} ^{j} -z_{t} ^{t} }{t-j} 
\label{GrindEQ__4_} 
\end{equation}

\section{Computation of probability of winning in unknown distribution case}
\label{sec:comp_pwin_unknown}
The probability of winning for a fixed value of $k$ and $T$ periods is

\noindent Thus,
\begin{align}
1&=\sum _{j=1}^{t-1}\frac{z_{t}^{j-t} -1}{t-j}  =\sum _{i=1}^{t-1}\frac{z_{t}^{-i} -1}{i} \label{GrindEQ__5_} \\
p_{t} (h)&=\left\{\begin{array}{cc} {\sum _{j=1}^{t}\frac{h^{j-1} -h^{t} }{t+1-j}  } & {z_{t} \le h} \\ {\frac{1-z_{t} ^{t} }{t} +\int _{h}^{z_{t} }p_{t-1} (x)dx +hp_{t-1} (h)} & {z_{t} >h} \end{array}\right. 
\end{align}

\[\int _{0}^{\infty }kF(x)^{k-1} f(x)\sum _{m=k+1}^{T}F(x)^{m-k-1}   \int _{x}^{\infty }f(y)F(y)^{T-m} dy dx\]
\[=\sum _{m=k+1}^{T}\int _{0}^{1}kx^{m-2}  \int _{x}^{1}y^{T-m} dy dx =\sum
_{m=k+1}^{T}\int _{0}^{1}kx^{m-2} \frac{1-x^{T-m+1} }{T-m+1}  dx \]
\[=\sum
    _{m=k+1}^{T}{\tfrac{k}{T-m+1}} \int _{0}^{1}x^{m-2} -x^{T-1}  dx
    =\sum
    _{m=k+1}^{T}{\tfrac{k}{T-m+1}} \left({\tfrac{1}{m-1}} -{\tfrac{1}{T}}
\right)  ={\tfrac{k}{T}} \sum _{m=k+1}^{T}{\tfrac{1}{m-1}}  \]

\section{Appendix Figures}

\begin{figure}
\centering
\includegraphics[width=0.95\textwidth]{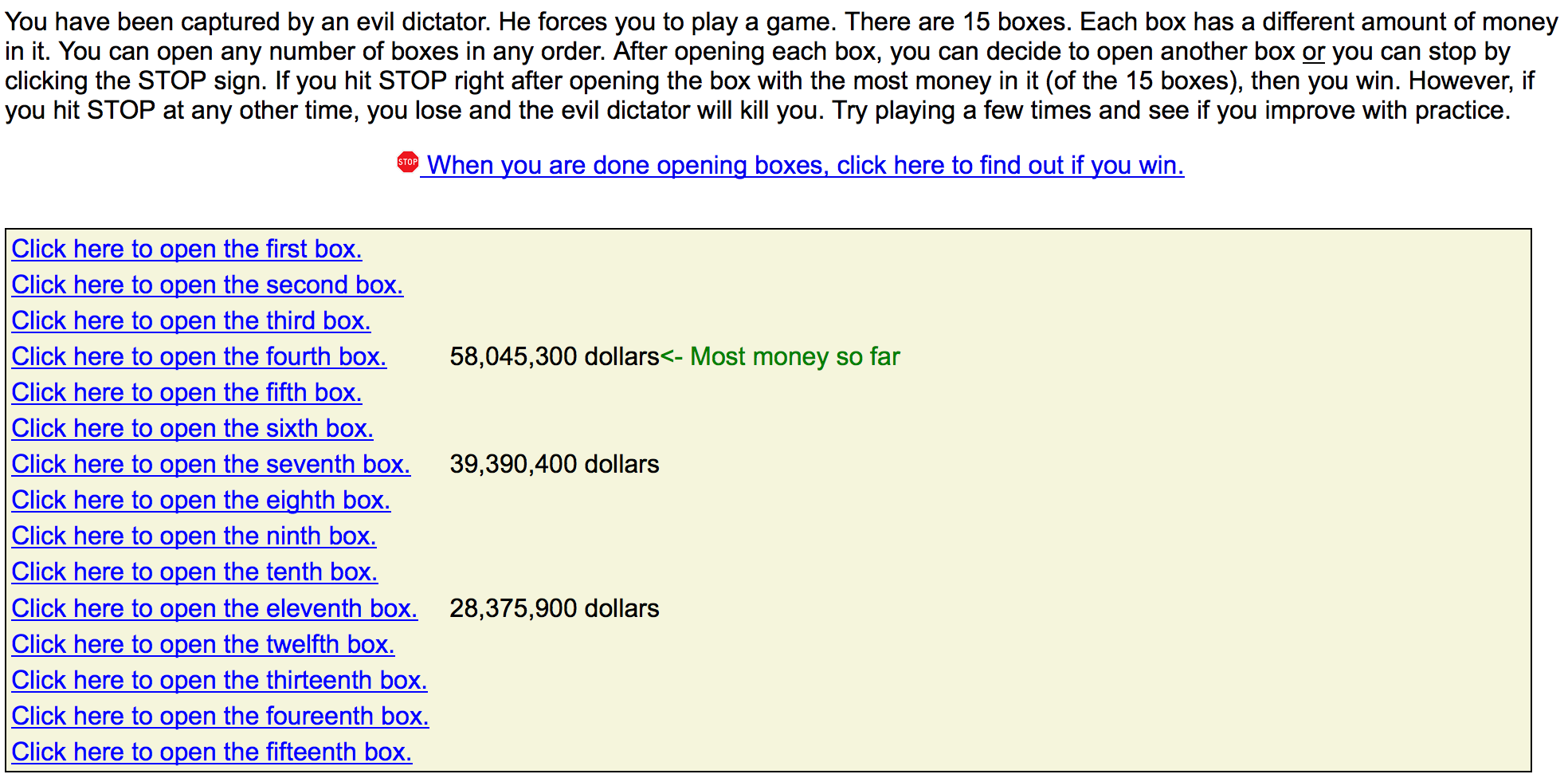}
\caption{Screenshot of the a 15 box treatment with 3 boxes opened.}
\label{fig:screenshot}
\end{figure}

\begin{figure}
\centering
\includegraphics[width=0.75\textwidth]{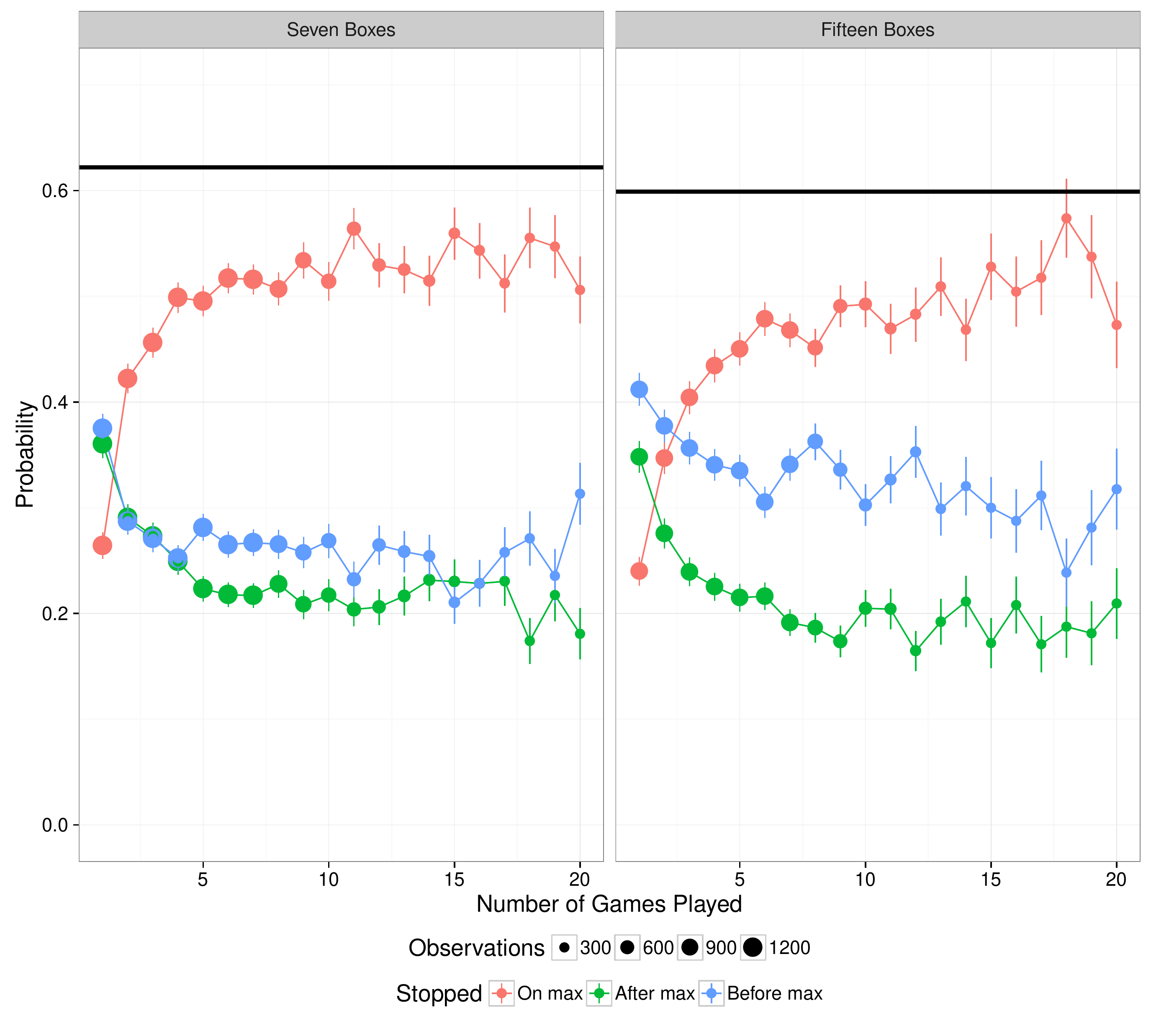}
\caption{Rates of winning the game for human players where each player
  played at least 7 games. Error
  bars indicated $\pm 1$ standard error, when they are not visible they are
smaller than the points.  The area of each point is proportional to the number of players in the average.}
\label{fig:prob_win7}
\end{figure}

\begin{figure}
\centering
\includegraphics[width=0.48\textwidth]{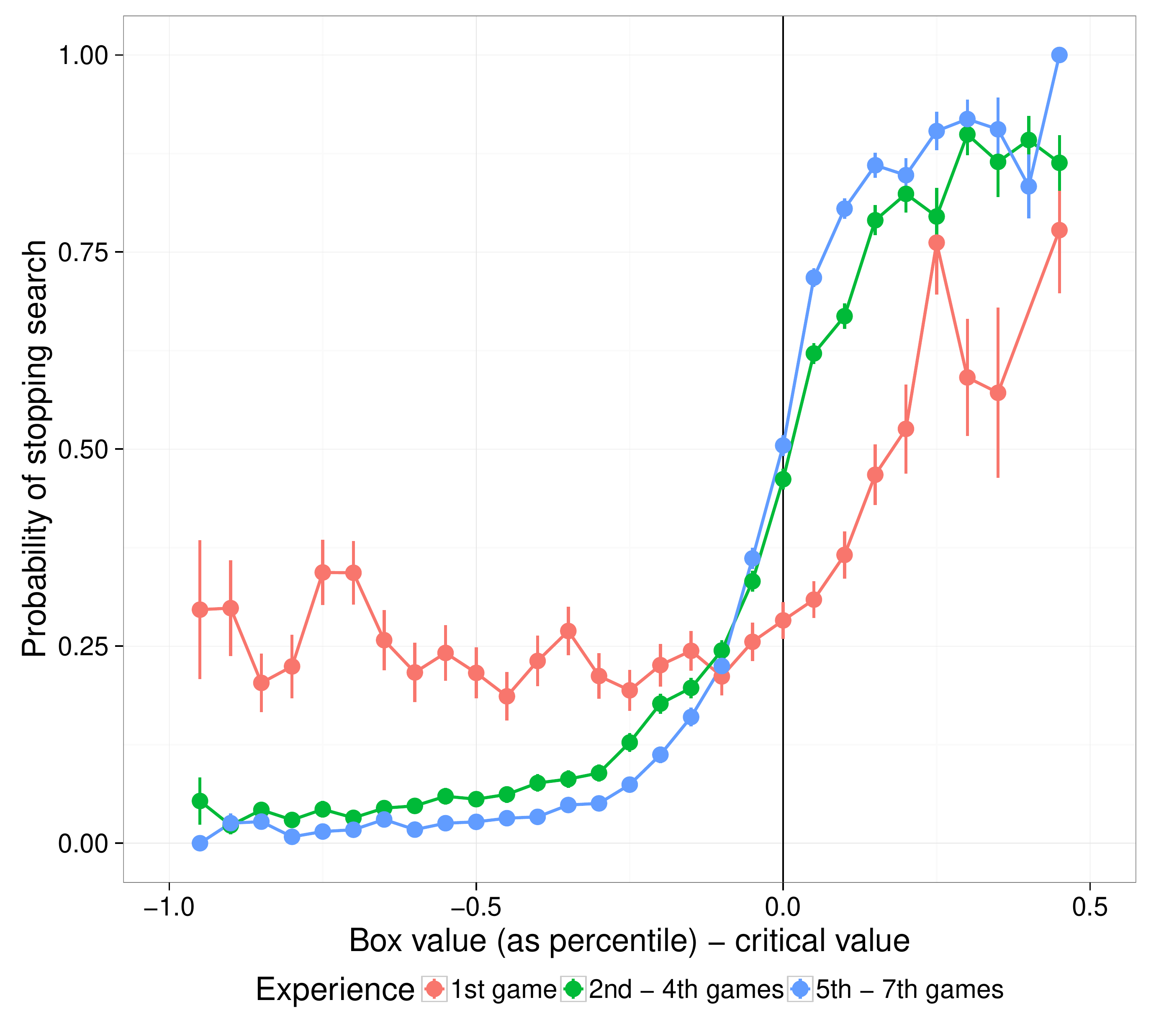}
\caption{ In Figure \ref{fig:prob_stop} different players contribute to
  different curves. For example, a player who only played one time would
  only contribute to the red curve, while someone who played 10 times would
  contribute to all three curves. To address these selection effects, in
  this plot, we restrict to the first 7 games of those who played at least 7 games.}
\label{fig:prob_stop_select}
\end{figure}

%% SS: We decided not to put the following figures in the appendix.

%% \begin{figure}
%% \centering
%% \includegraphics[width=0.90\textwidth]{prob_stop_vs_percentile_noreferrals_sigmoid_6panel.pdf}
%% \caption{}
%% \label{}
%% \end{figure}

%% \begin{figure}
%% \centering
%% \includegraphics[width=0.90\textwidth]{error_rates.pdf}
%% \caption{}
%% \label{}
%% \end{figure}

%% \begin{figure}
%% \centering
%% \includegraphics[width=0.90\textwidth]{search_noreferrals_depth_6cond.pdf}
%% \caption{}
%% \label{}
%% \end{figure}